\documentclass[runningheads]{llncs}

\usepackage[utf8]{inputenc}
\usepackage{preamble}

\title{Vector Commitments with Efficient Updates}

\usepackage{etoolbox}

\author{Ertem Nusret Tas \and Dan Boneh}
\institute{Stanford University}

\begin{document}

\maketitle

\begin{abstract}
Dynamic vector commitments that enable local updates of opening proofs 
have applications ranging from verifiable databases with membership changes to stateless clients on blockchains.
In these applications, each user maintains a relevant subset of the committed messages and the corresponding opening proofs with the goal of ensuring a succinct global state.
When the messages are updated, users are given some global update information and update their opening proofs to match the new vector commitment.
We investigate the relation between the size of the update information and the runtime complexity needed to update an individual opening proof.
Existing vector commitment schemes require that either the information size or the runtime scale {\em linearly} in the number~$k$ of updated state elements.
We construct a vector commitment scheme that asymptotically achieves both length and runtime that is {\em sublinear} in $k$, namely $k^\nu$ and $k^{1-\nu}$ for any $\nu \in (0,1)$.
We prove an information-theoretic lower bound on the relation between the update information size and runtime complexity that shows the asymptotic optimality of our scheme.
For $\nu = 1/2$, our constructions outperform Verkle commitments by about a factor of $2$ in terms of both the update information size and runtime, but makes use of larger public parameters.
\end{abstract}

\section{Introduction}
\label{sec:introduction}

A Vector Commitment (VC) scheme~\cite{cfm08,ly10,cf13} enables a committer to succinctly commit to a vector of elements.
Later, the committer can generate an \emph{opening proof} to prove that 
a particular position in the committed vector is equal to a certain value.
VCs have found many applications in databases and blockchains \cite{bitcoin,ethereum} as they enable 
a storage system to only store a commitment to the vector instead of the entire vector. 
The data itself can be stored elsewhere along with opening proofs. 
In a multiuser system, every user might store only one position of the vector along with the opening proof for that position. 

Dynamic VCs~\cite{cf13} are vector commitments that support updates to the vector. 
Suppose the committed vector is of length $N$ and some $k < N$ positions in the vector are updated,
so that a new vector commitment is published.
Then, every user in the system will need to update their local opening proof to match the updated commitment,
and this is done with the help of some global \emph{update information} $U$ that is broadcast to all users. 
This information is typically generated and published by a manager who maintains the entire vector.
Applications of dynamic VCs include verifiable databases, zero-knowledge sets with frequent updates~\cite{cf13} and stateless clients for blockchains~\cite{stateless-clients}.
The challenge is to design a VC scheme that minimizes the size of the update information $U$
as well as the computation work by each user to update their local opening proof. 

For example, consider stateless clients on a blockchain as an important application for dynamic VCs.
The state of the chain can be represented as a vector of length $N$, where position~$i$ corresponds to the state of account number~$i$.
Every user will locally maintain its own state (corresponding to some position in the vector) 
along with an opening proof that enables the user to convince a third party
as to its current state.
Whenever a new block is published, the state of the chain changes.  
In particular, suppose $k$ out of the $N$ positions in the vector need to be updated.  
The block proposer will publish the {\em update information} $U$ along with the new block,
and every user will update their opening proof to match the new committed state of the chain. 
Thus, users can ensure that their opening proofs are up to date with respect to the latest committed state of the chain.

We stress that in this application, the data being updated, namely the updated positions and diffs, is published as part of the block.
The update information $U$ only contains additional information that is needed to update the opening proofs.
When we refer to the size of $U$, we refer to its size, excluding the updated data (i.e., excluding the updated positions and diffs).

\begin{table}[]
\centering
\begin{tabular}{|c|c|c|c|}
 \hline
 Vector Commitment & $|U|$ & $T$ & PP \\  [0.5ex]
 \hhline{|=|=|=|=|}
 Merkle tree~\cite{merkle} & $\Tilde{\Theta}(k)\ |H|$ & $\Tilde{O}(1)$ & N \\
 \hline
 Verkle tree~\cite{verkle} & $\Tilde{\Theta}(k)\ |G|$ & $\Tilde{O}(1)\ |H|\ T_G$ & Y \\
 \hhline{|=|=|=|=|}
 \rule[-1em]{0pt}{2.5em} 
 This work with $\nu \in [0,1]$
 & $\Tilde{\Theta}(k^\nu) |H|$ & $\Tilde{\Theta}(k^{1-\nu})\ T_f$ & N \\
 \hhline{|=|=|=|=|}
 KZG commitments~\cite{kate} & $\Tilde{O}(1)$ & $\Tilde{\Theta}(k)\ T_G$ & Y \\
 \hline
 RSA accumulators and VCs~\cite{rsa-acc,rsa-vc}
 & $\Tilde{O}(1)$ & $\Tilde{\Theta}(k)\ T_G$ & N \\
 \hline
 Bilinear accumulators~\cite{bilinear-acc,bilinear-VC}
 & $\Tilde{O}(1)$ & $\Tilde{\Theta}(k)\ T_G$ & N \\
 \hline
\end{tabular}
\vspace{0.5cm}
\caption{Comparison of different VCs. $|U|$ denotes the length of the update information. $T$ denotes the runtime of a single proof update. 
$|G|$ and $|H|$ denote the size of a single group element and a single hash value, respectively. 
$T_G$ and $T_f$ denote the time complexity of a single group operation and a single function evaluation for the hash function used by the VC.
The last column PP is `Y' if the proof update requires pre-processing to generate a global and fixed table of auxiliary data needed for proof updates.
}
\vspace{-0.5cm}
\label{tab:vc-schemes}
\end{table}

In this paper, we investigate the trade-off between the length $|U|$ of the update information and the time complexity of proof updates.
Dynamic VCs can be grouped into two categories in terms of these parameters (Table~\ref{tab:vc-schemes}).
Tree-based VCs~\cite{merkle,hyperproofs} enable users to update their proofs in time $O(\polylog{N})$.
Each opening proof typically consists of $\polylog{(N)}$ inner nodes, 
and the update information $U$ contains the changes in the inner nodes affected by the message updates.
Each user calculates its new opening proof by downloading the relevant inner nodes published as part of $U$.
When $k$ positions are updated, a total of $O(k \log{(N)})$ inner nodes in the tree are affected in the worst case.
Thus, when each inner node has length $\Theta(\lambda)$, proportional to the security parameter $\lambda$, 
the update information consists of $O(k \log{(N)} \lambda)$ bits.

In contrast, algebraic VCs~\cite{kate,rsa-acc,rsa-vc,bilinear-acc,bilinear-VC} enable users to update their opening proofs with only 
knowledge of the updated data.
They do not require any additional update information $U$ to be published beyond the indices and the `diffs' of the updated data.
Thus, the length of the update information needed to update the opening proofs is $O(1)$. 
However, algebraic VCs typically require each user to read all of the changed messages and incorporate the effect of these changes on their proofs, resulting in $\Theta(k)$ work per proof update.

To summarize, while tree-based VCs support efficient calculation of the new opening proofs by publishing a large amount of update information, linear in $k$,
algebraic VCs do not require any additional update information beyond the updated data, but suffer from a large runtime for proof updates, linear in $k$.
We formalize the dichotomy of VCs in Section~\ref{sec:problem-statement}.

\subsection{Our Results}
\label{sec:introduction-construction}

We propose a family of VCs that can support \emph{sublinear update}, where both the length $|U|$ of the update information and the complexity of proof updates are sublinear in $k$.
More specifically, our VCs can attain $|U| = \Theta(k^\nu \lambda)$, $\nu \in (0,1)$, with a proof update complexity of $\Theta(k^{1-\nu})$ operations.
Our candidate construction with sublinear update is a \emph{homomorphic tree}, based on the VCs from~\cite{ACZ+20,hyperproofs}, where each inner node can be expressed as a \emph{sum} of the \emph{partial digests} of the messages underneath the node and its sibling (Section~\ref{sec:sublinear-complexity-efficient}).
We also propose an alternative construction called the \emph{homomorphic Merkle tree}, first formalized by~\cite{PT11,PSTY13,QZC+14} and based on the VCs from~\cite{PT11,PSTY13,QZC+14}.
The algebraic structure of these trees enable each user to calculate the \emph{effect} of a message update on any inner node without reading other inner nodes or messages.
The homomorphic tree constructions are based on pairings, and the homomorphic Merkle tree constructions are based on lattices~\cite{PT11,PSTY13,QZC+14,ACZ+20,hyperproofs}.

In Sections~\ref{sec:sublinear-complexity-efficient} and~\ref{sec:sublinear-complexity}, we provide update algorithms (Alg.~\ref{alg.vc.sublinear.p} and~\ref{alg.vc.sublinear}) for homomorphic Merkle trees and VCs based on homomorphic trees, parameterized by $\nu \in (0,1)$.
Our algorithm identifies a special subset of size $\Tilde{\Theta}(k^\nu)$ of the inner nodes affected by the message updates, and publish their new values as $U$; so that the users need not calculate these values.
These inner nodes are selected carefully to ensure that any inner node \emph{outside} of $U$ is affected by at most $\Theta(k^{1-\nu})$ updated messages.
Thus, to modify its opening proof, each user has to calculate the partial digests of at most $\Theta(k^{1-\nu})$ updated messages per inner node within its proof (that consists of $\Theta(\log{(N)})$ inner nodes).
Moreover, to calculate these partial digests, the user only needs the `diffs' of the updated messages.
This brings the asymptotic complexity of proof updates to $\Tilde{\Theta}(k^{1-\nu})$ operations, while achieving an update information size of $\Tilde{\Theta}(k^\nu \lambda)$ as opposed to $\Tilde{\Theta}(k\lambda)$ on Merkle trees using SHA256.

In Sections~\ref{sec:eval-homomorphic-tree} and~\ref{sec:eval-homomorphic}, we give estimates for the update information size and the update time of opening proofs for homomorphic tree and homomorphic Merkle tree constructions.
When the trade-off parameter $\nu$ is $1/2$, the homomorphic tree construction based on~\cite{ACZ+20} requires an update information size of $52.53$ kBytes and an update time of $0.34$ seconds.
In contrast, even the most efficient homomorphic Merkle tree construction~\cite{QZC+14} requires an update information size of $110.88$ MBytes and an update time of $32.6$ seconds when $\nu = 1/2$ (\cf Section~\ref{sec:eval-homomorphic}).
The large update information size is due to the lattice-based construction of the VCs from~\cite{PT11,PSTY13,QZC+14} used in the homomorphic Merkle trees.
Although the homomorphic tree constructions based on~\cite{ACZ+20,hyperproofs} have the same asymptotic performance as homomorphic Merkle trees, they are concretely more efficient for real-world parameters.
However, unlike homomorphic Merkle trees, the homomorphic tree constructions are not secure against quantum computers, require trusted setup, and have large public parameter sizes.

In Section~\ref{sec:lower-bound}, we prove an information theoretic lower bound on the size of the update information given an upper bound on the runtime complexity of proof updates.
The bound implies the asymptotic optimality of our scheme with sublinear update.
Its proof is based on the observation that if the runtime complexity is bounded by $O(k^{1-\nu})$, a user that wants to update its proof cannot read beyond $O(k^{1-\nu})$ updated messages.
Then, to calculate the effect of the remaining $k-O(k^{1-\nu})$ messages on its opening proof, the user has to download parts of the structured update information $U$.
Finally, to obtain the lower bound on $|U|$, we use Shannon entropy and lower bound the number of bits, namely $O(k^\nu \lambda)$, required to capture the total information that will be downloaded by the users; while maintaining the security of the VC with parameter $\lambda$.

\subsection{Applications}
\label{sec:introduction-applications}

We identify three main applications for VCs with sublinear update.

\subsubsection{Stateless clients for Ethereum}
Ethereum is the largest decentralized general purpose computation platform by market cap.
Ethereum state (\eg, user accounts) is currently stored in the form of a Merkle tree~\cite{eth-state-trie} and grows approximately by half every year~\cite{stateless-roadmap}.
Stateless clients~\cite{stateless-clients,stateless-roadmap} were proposed to mitigate the problem of state bloat and prevent the state storage and maintenance from becoming a bottleneck for decentralization.
Stateless clients maintain an opening proof to their account balances within the Ethereum state, thus can effortlessly prove the inclusion of their accounts within the latest state.
This enables the other Ethereum clients to verify the transactions that come with opening proofs without having to download the full state and check the validity of the claimed account balances.
Since block verification now requires downloading the proofs for the relevant state elements, Verkle trees~\cite{verkle-origin,verkle,verkle-aggregation} were proposed as a replacement for Merkle trees due to their short proof size.

Each new Ethereum block contains transactions that update the state elements and their opening proofs.
Archival nodes and block producers still maintain the full state so that they can inform the stateless clients about their new opening proofs \cite{stateless-roadmap}.
For this purpose, block producers must broadcast enough information to the clients over the peer-to-peer gossip network of Ethereum\footnote{Block producers can enable the clients to succinctly verify the correctness of this information via SNARK proofs, thus still keeping the verification cost of blocks small.}.
As minimizing the proof size was paramount to decentralizing verification for blocks, minimizing the update information size becomes necessary for decentralizing the role of the block producer who has to disseminate this information.
However, reducing the length of the update information must not compromise the low overhead of stateless clients by requiring larger number of operations per proof update.
Therefore, the ideal VC scheme for stateless clients must strike a delicate balance between the size of the update information and the runtime complexity of proof updates.

In Section~\ref{sec:verkle-update}, we provide the update algorithms 
for Verkle trees given their role in supporting stateless clients.
We observe that Verkle trees do not support sublinear update, and fall under the same category as tree-based VCs with update information length $\Tilde{\Theta}(k \lambda)$.
Despite this fact, Verkle trees are highly practical in terms of updates.
In Section~\ref{sec:evaluation}, we estimate that the update information size after a typical Ethereum block does not exceed $|U| \approx 100$ kBytes (compared to the typical block size of $<125$ kBytes).
Moreover, each Verkle proof can be updated within approximately less than a second on commodity hardware.
Thus, Verkle trees outperform homomorphic Merkle tree constructions, but fall below those based on homomorphic trees.
However, they have a smaller public parameter size (a few MBytes) compared to homomorphic tree constructions (a few GBytes).

\subsubsection{Databases with frequent membership changes}

VCs with sublinear update can support databases with frequent membership changes.
When a user first registers, a message is updated to record the membership of the user.
The user receives this record and its opening proof, using which it can later anonymously prove its membership.
When the user leaves the system, the message is once again updated to delete the record.
In all these steps, membership changes result in updates to the opening proofs of other members.
When these changes are frequent, it becomes infeasible to distribute new proofs after each change.
VCs with sublinear update offer an alternative and efficient way to update the opening proofs of the users in the event of such changes.

\subsection{Related Work}
\label{sec:related-work}
There are many VC constructions, each with different guarantees regarding the proof, commitment and public parameter sizes, verification time, updateability and support for subvector openings~\cite{cfm08,ly10,cf13,TXN20,PSTY13,LRY16,LM19,JS15,GRW+20,AR20,YLF+21,rsa-vc,bilinear-VC,hyperproofs,kate,verkle} (cf~\cite{vc-sok} for an SoK of VCs).
First formalized by~\cite{cf13}, almost all VCs allow some degree of updatability.
Whereas \cite{PSTY13,AR20,rsa-vc,bilinear-VC} enable updating the commitment and the opening proofs with only the knowledge of the old and the new messages, most VCs require some structured update information beyond the messages when the users do not have access to the internal data structures.
Among the lattice-based accumulators, vector commitments and functional commitments~\cite{JS15,LLN+16,PPS21,PSTY13,QZC+14,PT11,WW22}, constructions amenable to sublinear update are presented in \cite{PT11,PSTY13,QZC+14,PPS21,ACZ+20,hyperproofs}.
Homomorphic trees are based on \cite{ACZ+20,hyperproofs,CampanelliNRZZ22}.
On the other hand, homomorphic Merkle trees were formalized and instantiated by \cite{PT11,PSTY13,QZC+14} in the context of streaming authenticated data structures and parallel online memory checking.
The construction presented in~\cite[Section 3.4]{PPS21} offers an alternative VC with sublinear update as it is not a Merkle tree, yet has the property that each inner node can be expressed as a \emph{sum} of the partial digests of individual messages.

Dynamic vector commitments have found extensive use in the context of authenticated data structures~\cite{Tamassia03}.
For instance, Hasic et al. introduce an RSA accumulator construction to achieve a dynamic authenticated dictionary with $O(n^{\epsilon})$ update and query time and a constant verification time~\cite{GTH02}.
On a similar theme, Papamanthou et al. give a tree-based hash table construction that can achieve $O(n^{\epsilon})$ update time and constant ($O(1/\epsilon)$) size proofs~\cite{PTT16}.
Akin to Verkle trees, their construction uses RSA accumulators organized in the form of a wide tree with degree $n^{\epsilon}$.
Given the multitude of authenticated data structures, Miller et al. developed a language for programming authenticated operations (\eg, binary search trees) by simply describing the data structure~\cite{MHKS14}.
Li et al. make use of the authenticated
multipoint evaluation trees (AMT) from \cite{ACZ+20} to provide authenticated storage for blockchains with significantly
reduced I/O amplifications.

An alternative design to support stateless clients is the agregatable subvector commitment (aSVC) scheme \cite{asvc}, which is a VC that enables aggregating multiple opening proofs into a succinct subvector proof.
In aSVC, each user can update its opening proof with the knowledge of the transactions in the blocks, and block producers to prove the validity of these transactions succinctly by aggregating the proofs submitted by the transacting users.
As the scheme is based on KZG commitments, no update information is needed, yet, the update time complexity is linear in the number of transactions per block.

In the light of applications such as authenticated data structures and stateless clients, Campanelli et al. identify homomorphism, proof aggregation and updatability, maintenance of the VC given new updates and reusability of the existing trusted setup as desired features for VCs~\cite{CampanelliNRZZ22}.
Inspired by \cite{LaiM19}, they introduce linear-map VCs (LVCs) which can be opened to linear maps of the committed vector, with unbounded aggregation and the other aforementioned properties, and show two pairing-based LVC constructions for inner products.
They also provide two schemes, one based on multivariate polynomials and the other univariate, that support an arbitrary memory/time trade-off for openings, implying a similar trade-off between memory and the time to update the VC data structures after changes in the committed vector.
The multivariate construction generalizes Hyperproofs~\cite{hyperproofs} by supporting trees of any degree with potentially shorter opening proofs and by having commitments for any LVC at the leaves in place of the committed vector elements.
The univariate case generalizes~\cite{ACZ+20} and enables the setup to be independent of the memory/time trade-off.
In contrast, our work offers a flexible trade-off between the size of the update information and the time to update an individual opening proof that is independent of the setup and the degree of the homomorphic (Merkle) tree.

For dynamic accumulators that support additions, deletions and membership proofs, Camacho and Hevia proved that after $k$ messages are deleted, $\Omega(k)$ bits of data must be published to update the proofs of the messages in the initial accumulated set~\cite[Theorem 1]{impossibility}.
Their lower bound is information-theoretic and follows from a compression argument.
Christ and Bonneau subsequently used a similar method to prove a lower bound on the global state size of a \emph{revocable proof system} abstraction \cite{revocable}.
As revocable proof systems can be implemented by dynamic accumulators and vector commitments, their lower bound generalizes to these primitives, \ie, after $k$ messages are updated in a dynamic VC, at least $\Omega(k)$ bits of data must be published to update the opening proofs (see Appendix~\ref{sec:appendix} for the proof).
They conclude that a stateless commitment scheme must either have a global state with linear size in the number of accounts, or require a near-linear rate of local proof updates.
In our work, we already assume a linear rate of local proof updates, \ie, after every Ethereum block or $k$ messages in our parameterization, and that the message updates are publicized by the blockchain.
We instead focus on the trade-off between the global structured update information size (beyond the published messages) and the runtime complexity of proof updates.

In terms of lower bounds, Tamassia and Triandopoulos show a $\Omega(\log{(n)})$ bound on various cost metrics (\eg, proof size, communication and computational cost of updates) for hierarchical data processing (HDP) problems with an input set of $n$ elements, where the computation can be described by a directed acyclic graph (DAG)~\cite{TT05}.
Their result follows from an analogy between these cost measures and searching an element in an ordered sequence of $n$ elements.
They also propose a tree DAG construction with nearly optimal cost metrics among all key-graph structures.
Our lower bound in Section~\ref{sec:lower-bound} applies to multiple updates (as opposed to a single update incurring $O(\log{(n)})$ computational cost) and considers the scaling of the update information as a function of the number of updates (denoted by $k$) for \emph{any} VC, not necessarily restricted to a hash-based DAG structure.

Although we assume a trusted operator maintaining the VC in this work, there has been extensive research on verifying the consistency of the responses obtained from a server for an outsourced computation, in particular for the setting of concurrently querying clients~\cite{CSS07,CO18}.
Making the update information itself verifiable remains as future work.

\section{Preliminaries}
\label{sec:preliminaries}

\subsection{Notation}
\label{sec:notation}

We denote the security parameter by $\lambda$.
An event is said to happen with \emph{negligible probability}, if its probability, as a function of $\lambda$, is $o(1/\lambda^{d})$ for all $d>0$. 
An event happens with \emph{overwhelming probability} if it happens except with negligible probability.

We denote the set $\{0,1,2,\ldots,N-1\}$ by $[N]$.
When $y = O(h(x) \polylog{(x)})$, we use the shorthand $y=\Tilde{O}(h(x))$ (similarly for $\Theta(.)$ and $\Tilde{\Theta}(.)$).
The function $H(.) \colon \mathcal{M} \to \{0,1\}^\lambda$ represents a collision-resistant hash function.
We denote the binary decomposition of an integer $x$ by $\bin(x)$, and for $c>2$, its base $c$ decomposition by $\bin_c(x)$.
A vector of $N$ elements $(n_0, \ldots, n_{N-1})$ is shown as $(n_i)_{i}$.
The notation $\mathbf{x}[i{:}j]$ denotes the substring starting at the $i^\text{th}$ index and ending at the $j^\text{th}$ index within the sequence $\mathbf{x}$.
The indicator function $1_{P}$ is equal to one if the predicate $P$ is true, otherwise, it is zero.
In the subsequent sections, $k$ will be used to denote the number of updated messages.

For a prime $p$, let $\mathbb{F}_p$ denote a finite field of size $p$.
We use $\mathbb{G}$ to denote a cyclic group of prime order $p$ with generator $g$.
The Lagrange basis polynomial for a given $x \in \mathbb{F}_p$ is denoted as $L_x(X)$:
\begin{IEEEeqnarray*}{C}
L_x(X) = \prod_{\substack{i \in \mathbb{F}_p\\i \neq x}} \frac{X-i}{x-i}
\end{IEEEeqnarray*}

We will use $|G|$ and $|H|$ to denote the maximum size of the bit representation of a single group element and a single hash value respectively.
We will use $T_G$ and $T_f$ to denote the time complexity of a single group operation and a single function evaluation for the hash functions in Section~\ref{sec:merkle-tree-homomorphic}.

\subsection{Vector Commitments}
\label{sec:vector-commitments}

A vector commitment (VC) represents a sequence of messages such that each message can be proven to be the one at its index via an \emph{opening proof}.
A dynamic vector commitment allows updating the commitment and the opening proofs with the help of an \emph{update information} when the committed messages are changed.

\begin{definition}[from~\cite{cf13}]
Dynamic (updateable) vector commitments can be described by the following algorithms:

\smallskip\noindent
$\textsc{KeyGen}(1^\lambda, N) \to pp \colon$ Given the security parameter $\lambda$ and the size $N=\poly(\lambda)$ of the committed vector, the key generation algorithm outputs public parameters $pp$, which implicitly define the message space $\mathcal{M}$.

\smallskip\noindent
$\textsc{Commit}_{pp}(m_0, \ldots, m_{N-1}) \to (C, \aux) \colon$ Given a sequence of $N$ messages in $\mathcal{M}$ and the public parameters $pp$, the commitment algorithm outputs a commitment string $C$ and the data $\aux$ required to produce the opening proofs for the messages.
Here, $\aux$ contains enough information about the current state of the VC's data structure (\ie, the current list of committed messages) to help generate the opening proofs.

\smallskip\noindent
$\textsc{Open}_{pp}(m, i, \aux) \to \pi_i \colon$ The opening algorithm is run by the committer to produce a proof $\pi_i$ that $m$ is the $i^\text{th}$ committed message. 

\smallskip\noindent
$\textsc{Verify}_{pp}(C, m, i, \pi_i) \to \{0,1\} \colon$ The verification algorithm accepts (\ie, outputs 1) or rejects a proof.  
The security definition will require that $\pi_i$ is accepted only if $C$ is a commitment to some $(m_0, \ldots, m_{N-1})$ such that $m = m_i$.

\smallskip\noindent
$\textsc{Update}_{pp}(C, (i, m_i)_{i \in [N]}, (i, m'_i)_{i \in [N]}, \aux) \to (C', U, \aux') \colon$ The algorithm is run by the committer to update the commitment $C$ when the messages $(m_{i_j})_{j \in [k]}$ at indices $(i_j)_{j \in [k]}$ are changed to $(m'_{i_j})_{j \in [k]}$.  The other messages in the vector are unchanged. It takes as input the old and the new messages, their indices and  the data variable $\aux$. It outputs a new commitment $C'$, update information $U$ and the new data variable $\aux'$.

\smallskip\noindent
$\textsc{ProofUpdate}_{pp}(C, p((i, m_i)_{i \in [N]}, (i, m'_i)_{i \in [N]}), \pi_j , m', i, U) \to \pi_j' \colon$ 
The proof update algorithm can be run by any user who holds a proof $\pi_j$ for some message at index $j$ and a (possibly) new message $m'$ at that index.
It allows the user to compute an updated proof $\pi'_j$ (and the updated commitment $C'$) such that $\pi'_j$ is valid with respect to $C'$, which contains $m'_i$, $i \in N$, as the new messages at the indices $ i \in N$ (and $m'$ as the new message at index $i$).
Here, $p(.)$ specifies what portion of the old and the new messages is sufficient to update the opening proof.
For instance, the proof update algorithm often does not need the old and the new messages in the open; but can carry out the proof update using only their differences.
In this case, $p((i, m_i)_{i \in [N]}, (i, m'_i)_{i \in [N]}) = (i, m'_i-m_i)_{i \in N}$.
\end{definition}

\emph{Correctness} of a VC requires that $\forall N = \poly(\lambda)$, for all honestly generated parameters $pp \xleftarrow{} \textsc{KeyGen}(1^\lambda, N)$, given a commitment $C$ to a vector of messages $(m_0, \ldots, m_{N-1}) \in \mathcal{M}^N$, generated by  $\textsc{Commit}_{pp}$ (and possibly followed by a sequence of updates), and an opening proof $\pi_i$ for a message at index $i$, generated by $\textsc{Open}_{pp}$ or $\textsc{ProofUpdate}_{pp}$, it holds that $\textsc{Verify}_{pp}(C, m_i, i, \pi_i)=1$ with overwhelming probability.

\emph{Security} of a VC is expressed by the position-binding property:
\begin{definition}[Definition 4 of~\cite{cf13}]
A VC satisfies position-binding if $\forall i \in [N]$ and for every PPT adversary $\mathcal{A}$, the following probability is negligible in $\lambda$:
\begin{IEEEeqnarray*}{C}
\Pr\left[\substack{\textsc{Verify}_{pp}(C, m, i, \pi_i) = 1 \land \\ \textsc{Verify}_{pp}(C, m', i, \pi'_i) = 1 \land m \neq m'}\ \colon \substack{pp \xleftarrow{} \textsc{KeyGen}(1^\lambda, N) \\ (C, m, m', \pi_i, \pi'_i) \xleftarrow{} \mathcal{A}(pp)}\right]
\end{IEEEeqnarray*}
\end{definition}

We relax the \emph{succinctness} assumption of~\cite{cf13} and denote a value to be succinct in $x$ if it is $\polylog(x)$.

Many VC constructions also satisfy the hiding property:
informally, no PPT adversary $\mathcal{A}$ should be able to distinguish whether the VC was calculated for a vector $(m_0, \ldots, m_{N-1})$ or a vector $(m'_0, \ldots, m'_{N-1}) \neq (m_0, \ldots, m_{N-1})$.
In this work, we do not consider the hiding property since it is not explicitly required by our applications, and VCs can be made hiding by combining them with a hiding commitment \cite{cf13}.

\subsection{KZG Polynomial Commitments}
\label{sec:kzg-commitments}

The KZG commitment scheme~\cite{kate} commits to polynomials of degree bounded by $\ell$ using the following algorithms:

\smallskip\noindent
$\textsc{KeyGen}(1^{\lambda}, \ell) \to pp \colon$ outputs $pp = (g, g^\tau, g^{(\tau^2)}, \ldots, g^{(\tau^\ell)})$ as the public parameters, where $g$ is the generator of the cyclic group $\mathbb{G}$ and $\tau$ is a trapdoor ($pp[i] = g^{\tau^i}$).

\smallskip\noindent
$\textsc{Commit}\bigl(pp, \phi(X)\bigr) \to (C, \aux) \colon$
The commitment to a polynomial $\phi(X) = \sum_{i=0}^{\ell-1} a_i X^i$ is denoted by $[\phi(X)]$, 
and is computed as $[\phi(X)] = \prod_{i=0}^\ell (pp[i])^{a_i}$.
The commitment algorithm outputs $C = [\phi(X)]$ and $\aux = \phi(X)$.

\smallskip\noindent
$\textsc{Open}_{pp}(m, i, \aux) \to \pi:$ 
outputs the opening proof $\pi_i$ that $\phi(i) = m$, calculated as the commitment to the quotient polynomial $(\phi(X)-\phi(i)) / (X-i)$.

\smallskip\noindent
$\textsc{Verify}(C, m, i, \pi)$ accepts if the pairing check 
$ e\left(C/g^m, g\right) = e\left(\pi,  pp[1]/g^i \right) $
holds.

\medskip\noindent
We refer to~\cite{kate} for the security analysis of this scheme.

\subsection{Merkle Trees}
\label{sec:merkle-trees}

Merkle Tree is a vector commitment using a collision-resistant hash function.
In a Merkle tree, hashes of the committed messages constitute the leaves of a $c$-ary tree of height $h = \log_c(N)$, where each inner node is found by hashing its children.
The depth of the root is set to be $0$ and the depth of the leaves is $\lceil \log_c(N) \rceil$.
The commitment function outputs the Merkle root as the commitment $C$ and the Merkle tree as $\aux$.
The opening proof for a message $m_x$ at some index $x$ is the sequence of $h(c-1)$ hashes consisting of the siblings of the inner nodes on the path from the root to the hash of the message $m_x$.
We hereafter consider binary Merkle trees ($c=2$) and assume $N=c^h = 2^h$ unless stated otherwise.
Let $u_{b_0, b_1,\ldots,b_{i-1}}$, $b_j \in \{0,1\}$, $j \in [i]$, denote an inner node at depth $i-1$ that is reached from the root by choosing the left child at depth $j$ if $b_j=0$ and the right child at depth $j$ if $b_j=1$ ($b_0=\bot$ and $u_{\bot}$ is the root). 
By definition, for a message $m_x$ at index $x$, $H(m_x) = u_{\bot,\bin(x)}$.

\subsection{Verkle Trees}
\label{sec:verkle-trees}

A Verkle tree~\cite{verkle,verkle-aggregation} is similar to a Merkle tree except that each inner node is calculated as the hash of the KZG polynomial commitment to its children.
Let $b_j \in [c]$, $j=1, \ldots, h$, denote the indices of the inner nodes on the path from the root to a leaf at index $x$, $\bin_c(x) = (b_1, \ldots, b_h)$, relative to their siblings.
Define $f_{b_0,\ldots,b_j}$, $j \in [h]$, as the polynomials determined by the children of the inner nodes on the path from the root to the leaf, where $f_{b_0}=f_\bot$ is the polynomial determined by the children of the root.
Let $C_{b_0,\ldots,b_j} = [f_{b_0,\ldots,b_j}]$, $j \in [h]$, denote the KZG commitments to these polynomials.
By definition, $u_{b_0,\ldots,b_j} = H(C_{b_0,\ldots,b_j})$, and the value of the polynomial $f_{b_0,\ldots,b_j}$ at index $b_{j+1}$ is $u_{b_0,\ldots,b_{j+1}}$ for each $j \in [h]$.
Here, $u_{b_0} = H(C_{b_0})$ is the root of the tree, and $u_{b_0,\ldots,b_h}$ equals the hash $H(m_x)$ of the message at index $x$.
For consistency, we define $C_{b_0,\ldots,b_h}$ as $m_x$.
For example, given $h = 3$ and $c = 4$, the inner nodes from the root to the message $m_{14}$ have the indices $b_0 = 0$, $b_1 = 3$ and $b_2 = 2$, and they are committed by the polynomials $f_{\bot}$, $f_{\bot,0}$ and $f_{\bot,0,3}$ respectively.

The commitment function $\textsc{Commit}_{pp}(m_0, \ldots, m_{N-1})$ outputs the root $u_{b_0}$ as the commitment $C$ and the Verkle tree itself as $\aux$.

The Verkle opening proof for the message $m_x$, $\bin(x) = (b_1, \ldots, b_h)$, consists of two parts: (i) the KZG commitments $(C_{b_0,b_1}, \ldots, C_{b_0,\ldots, b_{h-1}})$ on the path from the root to the message, and (ii) a Verkle multiproof.
The goal of the Verkle multiproof is to show that the following evaluations hold for the inner nodes from the root to the message: $f_{b_0,\ldots,b_j}(b_{j+1})=u_{b_0,\ldots,b_{j+1}} = H(C_{b_0,\ldots,b_{j+1}})$, $j \in [h]$.
It has two components: (i) the commitment $[g(X)]$ and (ii) the opening proof $\pi'$ for the polynomial $h(X)-g(X)$ at the point $t=H(r,[g(X)])$, where
\begin{IEEEeqnarray*}{C}
g(X)=\sum_{j=0}^{h-1} r^j \frac{f_{b_0,\ldots,b_j}(X)-u_{b_0,\ldots,b_{j+1}}}{X-b_{j+1}}, \quad\ h(X)=\sum_{j=0}^{h-1} r^j \frac{f_{b_0,\ldots,b_{j}}(X)}{t-b_{j+1}},
\end{IEEEeqnarray*}
and $r=H(C_{b_0},..,C_{b_0,\ldots,b_{h-1}},u_{b_0,b_1},..,u_{b_0,\ldots,b_h},b_1,..,b_h)$.
Thus, $\textsc{Open}_{pp}(m, i, \aux)$ outputs $((C_{b_0,b_1}, \ldots, C_{b_0,\ldots,b_{h-1}}), ([g(X)], \pi'))$.

To verify a Verkle proof $\pi = ((C_{b_0,b_1}, \ldots, C_{b_0,\ldots,b_{h}}), (D,\pi'))$, the verification algorithm $\textsc{Verify}_{pp}(C, m, x, \pi)$ first computes $r$ and $t$ using $u_{b_0,\ldots,b_j} = H(C_{b_0,\ldots,b_j})$, $j \in [h]$, and $u_{b_0,\ldots,b_h} = H(m)$.
Then, given the indices $\bin(x) = (b_1, \ldots, b_h)$ and the commitments $(C_{b_0,b_1}, \ldots, C_{b_0,\ldots,b_{h}})$, it calculates
\begin{IEEEeqnarray*}{C}
y = \sum_{j=0}^{h-1} r^j \frac{C_{b_0,\ldots,b_j}}{t-b_{j+1}} \quad\quad\quad E = \sum_{j=0}^{h-1} \frac{r^j}{t-b_{j+1}} C_{b_0,\ldots,b_j}.
\end{IEEEeqnarray*}
Finally, it returns true if the pairing check $e(E-D-[g(X)],[1]) = e(\pi', [X-t])$ is satisfied.

As the degree $c$ of a Verkle tree increases, size of the opening proofs and the runtime of the verification function decreases in proportion to the height $h = \log_c{N}$ of the tree.
This enables Verkle trees to achieve a short opening proof size for large number of messages (as in the case of the Ethereum state trie) by adopting a large degree (\eg, $c=256$).
In comparison, each Merkle proof consists of $(c-1) \log_c{N}$ inner nodes, which grows linearly as $c$ increases.

\section{Formalizing the Dichotomy of VCs}
\label{sec:problem-statement}

We first analyze the trade-off between the number of operations required by proof updates and the size of the update information $U$ by inspecting different types of dynamic VCs.
Recall that the number of updated messages is $k \leq N$.

\subsection{Updating KZG Commitments and Opening Proofs}
\label{sec:kzg-update}

In the subsequent sections, we assume that each user has access to a dictionary of KZG commitments to the Lagrange basis polynomials $L_i(X)$, $i \in \mathbb{F}_p$, and for each polynomial, its opening proofs at each point $j \in \mathbb{F}_p$, $j < N$.
With the help of this table, one can instantiate a KZG based VC to the messages $(m_i)_{i \in [N]}$, by treating them as the values of the degree $N$ polynomial $\phi(X)$ at inputs $i \in \mathbb{F}_p$, $i<N$.
We next analyze the complexity of the update information and the proof updates in this VC.
The update and proof update algorithms are described by Alg.~\ref{alg.kzg} in Appendix~\ref{sec:appendix-algorithms}.

\subsubsection{Update Information}
Suppose the vector $(i, m_i)_{i \in [N]}$ is updated at some index $i$ such that $m'_i \xleftarrow{} m_i + \delta$ for some $\delta \in \mathbb{F}_p$.
Then, the polynomial $\phi(X)$ representing the vector is replaced by $\phi'(X)$ such that $\phi'(X) = \phi(X)$ if $X \neq i$, and $\phi'(i) = \phi(i) + \delta$ at $X = i$. 
Thus, the new KZG commitment $C'$ to $\phi'(X)$ is constructed from the commitment $C$ to $\phi(X)$ as follows:
\begin{IEEEeqnarray*}{C}
[\phi'(X)] = [\phi(X)+\delta L_i(X)] = [\phi(X)][L_i(X)]^{\delta}
= C \cdot [L_i(X)]^{\delta} = C \cdot [L_i(X)]^{m'_i-m_i}
\end{IEEEeqnarray*}
If the vector is modified at $k$ different indices $i_1,...,i_k$ from message $m_{i_j}$ to $m'_{i_j}$, $j \in [k]$, then
the new commitment $C' = [\phi'(X)]$ becomes
\begin{IEEEeqnarray*}{rCl}
\left[\phi(X)+\sum_{j=1}^k (m'_{i_j}-m_{i_j}) L_{x_{i_j}}(X)\right] &=& [\phi(X)] \prod_{j=1}^k[L_{i_j}(X)]^{(m'_{i_j}-m_{i_j})} \\
&=& C \prod_{j=1}^k[L_{i_j}(X)]^{(m'_{i_j}-m_{i_j})}.
\end{IEEEeqnarray*}
Thus, the commitment can updated given only the old and the new messages at the updated indices, besides the table.

\subsubsection{Proof Update}
Let $\pi_{x}$ denote the opening proof of a polynomial $\phi(X)$ at a point $(x,m_{x})$.
When $k$ messages are updated, the new opening proof $\pi'_{x}$ can be found as a function of the old proof $\pi_{x}$ and the opening proofs $\pi_{i_j,x}$ of the Lagrange basis polynomials $L_{i_j}(X)$, $j \in [k]$, at the index $x$ ($m'_{x} = m_{x}+\sum_{j=1}^k (m'_{i_j}-m_{i_j}) \cdot 1_{x=i_j}$ is the new value of $m_{x}$ after the $k$ updates):
\begin{IEEEeqnarray*}{rCl}
\pi'_x &=& \left[\frac{\phi'(X)-m_{x}-\sum_{j=1}^k \delta_j \cdot 1_{x=i_j}}{X-x}\right] \\
&=& \pi_{x} \prod_{j=1}^k \left[\frac{L_{i_j}(X)-L_{i_j}(x)}{X-x}\right]^{m'_{i_j}-m_{i_j}}
= \pi_x \prod_{j=1}^k \pi^{m'_{i_j}-m_{i_j}}_{i_j,x}
\end{IEEEeqnarray*}
Thus, the proof can updated given only the old and the new messages at the updated indices, besides the table.
The update information is set to be the empty set, \ie, $U = \emptyset$.

\subsubsection{Complexity}
The size of the update information is constant, \ie, $\Tilde{\Theta}(1)$.
Each user can update its proof after $k$ accesses to the dictionary, and in the worst case, $\Theta(k \log{|\mathcal{M}|}) = \Tilde{\Theta}(k)$ group operations as $\log{(m'_i-m_i)} \leq \log{|\mathcal{M}|}$ for all $i \in [N]$.

\subsection{Updating Merkle Trees and Opening Proofs}
\label{sec:merkle-update}

We next consider a Merkle tree and analyze the complexity of the update information size and the runtime for proof updates.
A simple update scheme would be recalculating the new Merkle tree given all of the old messages or the old inner nodes of the Merkle tree, and the message updates.
However, this implies a large complexity for the runtime of the proof update algorithm that scales as $\Omega(k)$ when users keep track of the inner nodes, and as $\Omega(N)$ when the users recalculate the tree from scratch at each batch of updates.
Moreover, in many applications, the users do not have access to any messages or inner nodes besides those that are part of the Merkle proof held by the user.
Hence, in the following sections, we describe update and proof update algorithms
that reduce the runtime complexity of the proof updates at the expanse of larger update information (Alg.~\ref{alg.merkle} in Appendix~\ref{sec:appendix-algorithms}).

\subsubsection{Update Information}
Suppose the vector $(i, m_i)_{i \in [N]}$ is updated at some index $x$, $(b_1,\ldots,b_h) = \bin(x)$, to $m'_x$.
Then, the root $C=u_{b_0}$ and the inner nodes $(u_{b_0,b_1}, \ldots, u_{b_0,b_1,\ldots,b_h})$, $(b_1,\ldots,b_h) = \bin(i)$, must be updated to reflect the change at that index.
Given the old inner nodes, the new values for the root and these inner nodes, denoted by $C'=u'_{b_0}$ and $(u'_{b_0,b_1}, \ldots, u'_{b_0,b_1,\ldots,b_h})$, are calculated recursively as follows:
\begin{IEEEeqnarray*}{rCl}
u'_{b_0,b_1,\ldots,b_h} &\xleftarrow[]{}& H(m'_x), \\
u'_{b_0,b_1,\ldots,b_j} &\xleftarrow[]{}^&
\begin{cases}
    H(u'_{b_0,b_1,\ldots,b_j,0}, u_{b_0,b_1,\ldots,b_j,1}) &\text{if } b_{j+1} = 0,\ j<h \\
    H(u_{b_0,b_1,\ldots,b_j,0}, u'_{b_0,b_1,\ldots,b_j,1}) &\text{if } b_{j+1} = 1,\ j<h \\
\end{cases}
\end{IEEEeqnarray*}
When the messages are modified at $k$ different points $i_j$, $j \in [k]$, the calculation above is repeated $k$ times for each update.

As the updated inner nodes are parts of the Merkle proofs, the update information consists of the new values at the inner nodes listed from the smallest to the largest depth in the canonical left to right order.
For instance,
$U = ((\bot, u'_{\bot}), (\bot0, u'_0), (\bot1, u'_1), (\bot00, u'_{00}), (\bot10, u'_{10}), \ldots)$
implies that the root $u_{\bot}$ and the inner nodes $u_{\bot0}$, $u_{\bot1}$, $u_{\bot00}$ and $u_{\bot10}$ were updated after $k$ messages were modified at the leaves of the Merkle tree.
We reference the updated inner nodes using their indices (\eg, $U[b_0, b_1 \ldots b_j] = v$, when $(b_1 \ldots b_j, v) \in U$).

\subsubsection{Proof Update}
The Merkle proof $\pi_x$ for a message at index 
$x$, $(b_1, \ldots, b_h) = \bin(x)$, is the sequence $(u_{\overline{b}_1}, u_{b_1,\overline{b}_2}, \ldots, u_{b_1,b_2,\ldots,\overline{b}_h})$.
When $k$ messages are updated, some of the inner nodes within the proof might have changed.
A user holding the Merkle proof for index $x$ can find the new values of these inner nodes by querying the update information with their indices.

\subsubsection{Complexity}
Upon receiving the update information $U$, each user can update its proof in $\Theta(\log^2{(N)}+|H| \log{(N)}) = \Tilde{\Theta}(1)$ time by running a binary search algorithm to find the updated inner nodes within $U$ that are part of its Merkle proof, and reading the new values at these nodes.
Since modifying each new message results in $h = \log{(N)}$ updates at the inner nodes and some of the updates overlap, $|U| = \Theta(k \log{(N/k)} (\log{(N)}+|H|)) = \Tilde{\Theta}(k)|H|$, as each updated inner node is represented by its index of size $\Theta(\log{(N)})$ and its new value of size $|H|$ in $U$.

\subsection{Dichotomy of VCs}
\label{sec:goal}

In the case of KZG commitments, $|U| = \Tilde{\Theta}(1)$, and there is no information overhead on top of the message updates.
For Merkle trees with an efficient proof update algorithm, $|U| = \Tilde{\Theta}(k)|H|$,
thus there is an extra term scaling in $\Tilde{\Theta}(k)|H| = \Tilde{\Theta}(k)\lambda$, since $|H| = \Omega(\lambda)$ for collision-resistant hash functions.
In contrast, for KZG commitments, each user has to do $\Tilde{\Theta}(k)$ group operations to update its opening proof; whereas in Merkle trees, each user can update its proof in $\Tilde{\Theta}(1)$ time, which does not depend on $k$.
Hence, KZG commitments outperform Merkle trees in terms of the update information size, whereas Merkle trees outperform KZG commitments in terms of the time complexity of proof updates.
Table~\ref{tab:vc-schemes} generalizes this observation to a dichotomy between algebraic VC schemes and tree-based ones favoring shorter runtimes for proof updates.
The algebraic and tree-based ones outperform each other in terms of the update information size and runtime complexity respectively.

\section{Vector Commitments with Sublinear Update}
\label{sec:sublinear-complexity-efficient}

We would like to resolve the separation in Table~\ref{tab:vc-schemes} and obtain a vector commitment, where both the size of the update information and the complexity of proof updates have a sublinear dependence on $k$.
In particular, $|U| = \Tilde{\Theta}(g_1(k)\lambda)$ in the worst case, 
and the proof update algorithm requires at most $\Tilde{\Theta}(g_2(k))$ operations, where both $g_1(k)$ and $g_2(k)$ are $o(k)$.
We say that such a VC supports \emph{sublinear update}.

In this section, we describe a family of VCs with sublinear update,parameterized by the values $\nu \in (0,1)$ and characterized by the functions $(g_1,g_2) = (k^{\nu}, k^{1-\nu})$.

\subsection{Homomorphic Trees}
\label{sec:tree-homomorphic}

We first introduce homomorphic trees where messages placed in the leaves take values in a set $\mathcal{M}$, and those at the inner nodes take value in a set $\mathcal{D}$.
Here, $\mathcal{M}$ is a finite field $\mathbb{F}_q$ of size $q$, and $\mathcal{D}$ is a cyclic group $\mathbb{G}$ of prime order $q$.
Define $\bin(i)[0:-p]$ as the empty string $\epsilon$ for $p \in \mathbb{Z}^+$.
The homomorphic property of these trees refers to the fact that there are efficiently computable functions
\[
   h_{i,j,c}: \mathcal{M} \to \mathcal{D}  \qquad \text{for $i \in [N]$ and $j \in [h]$ and $c \in \{0,\ldots,2^p-1\}$}
\]
parameterized by some $p \in \mathbb{Z}^+$, such that every inner node $u_{b_0,\ldots,b_j} \in \mathcal{D}$ can be expressed as
\begin{IEEEeqnarray*}{llCl}
\text{for the root: } & u_{b_0} &=& \hspace{2.5em} \sum_{i \in [N]} \hspace{0.5em} h_{i,0,\epsilon}(m_i) \\*
\text{for the inner nodes: } & u_{b_0,\ldots,b_j} &=& \hspace{-1em} \sum_{i\colon\bin(i)[0{:}j-p-1]=(b_1,\ldots,b_{j-p})} \hspace{-3em} h_{i,j,(b_{j-p+1},\ldots,b_j)}(m_i).
\end{IEEEeqnarray*}
Here, the parameter $p$ determines the dependence of the inner nodes on the underlying messages.
When $p = 0$, each inner node is a homomorphic function of the messages under the node.
When $p = 1$, each inner node is a homomorphic function of the messages under its parent, \ie, under both the node itself and its sibling.
We refer to the function $h_{i,j,(b_{j-p+1},\ldots,b_j)}(.)$ as a \emph{partial digest function} and refer to $h_{i,j,(b_{j-p+1},\ldots,b_j)}(m_i)$ as the \emph{partial digest} of $m_i$ on the $(b_{j-p+1},\ldots,b_j)$-th descendant at depth $j$ (measured from the root) of the node at depth $j-p$ that is on the path from $m_i$ to the root.
In other words, $h_{i,j,(b_{j-p+1},\ldots,b_j)}(m_i)$ denotes the effect of the message $m_i$ on the sibling and `cousins' of the node at depth $j$ that lay on the path from $m_i$ to the root.
Let $h$ denote the height of the homomorphic tree.
(Note that the equations above use the additive notation for the group operation. 
When applied to group elements, `$+$' denotes the group operation, and when applied to field elements, it denotes the field's addition operation.) 

As an example, consider a homomorphic tree for a VC that commits to four messages $m_0,m_1,m_2,m_3$.
Suppose $p = 0$.
Then, its root $u_{\bot}$ and inner nodes $u_{\bot,0}$, $u_{\bot,1}$, $u_{\bot,0,0}$, $u_{\bot,0,1}$, $u_{\bot,1,0}$, $u_{\bot,1,1}$ can be calculated as follows:
\begin{IEEEeqnarray*}{rClrCl}
u_{\bot} &=& h_{0,0,\epsilon}(m_0) + h_{1,0,\epsilon}(m_1) + h_{2,0,\epsilon}(m_2) + h_{3,0,\epsilon}(m_3)\ ; \qquad   &    u_{\bot,0,0} &=& h_{0,2,\epsilon}(m_0) \\*
u_{\bot,0} &=& h_{0,1,\epsilon}(m_0) + h_{1,1,\epsilon}(m_1)\ ;                                      &    u_{\bot,0,1} &=& h_{1,2,\epsilon}(m_1) \\*
u_{\bot,1} &=& h_{2,1,\epsilon}(m_2) + h_{3,1,\epsilon}(m_3)\ ;                                      &    u_{\bot,1,0} &=& h_{2,2,\epsilon}(m_2) \\*
&&                                                                              &    u_{\bot,1,1} &=& h_{3,2,\epsilon}(m_3)
\end{IEEEeqnarray*}
On the other hand, when $p = 1$, the root $u_{\bot}$ and inner nodes $u_{\bot,0}$, $u_{\bot,1}$, $u_{\bot,0,0}$, $u_{\bot,0,1}$, $u_{\bot,1,0}$, $u_{\bot,1,1}$ can be calculated as follows:
\begin{IEEEeqnarray*}{rCl}
u_{\bot} &=& h_{0,0,\epsilon}(m_0) + h_{1,0,\epsilon}(m_1) + h_{2,0,\epsilon}(m_2) + h_{3,0,\epsilon}(m_3) \\*
u_{\bot,0} &=& h_{0,1,0}(m_0) + h_{1,1,0}(m_1) + h_{2,1,0}(m_2) + h_{3,1,0}(m_3) \\*
u_{\bot,1} &=& h_{0,1,1}(m_0) + h_{1,1,1}(m_1) + h_{2,1,1}(m_2) + h_{3,1,1}(m_3) \\*
\end{IEEEeqnarray*}
\vspace{-1cm}
\begin{IEEEeqnarray*}{rClrCl}
u_{\bot,0,0} &=& h_{0,2,0}(m_0) + h_{1,2,0}(m_1)\ ;\quad                                      &    u_{\bot,0,1} = h_{0,2,1}(m_0) + h_{1,2,1}(m_1) &&\\*
u_{\bot,1,0} &=& h_{2,2,0}(m_2) + h_{3,2,0}(m_3)\ ;\quad                                      &    u_{\bot,1,1} = h_{2,2,1}(m_2) + h_{3,2,1}(m_3) &&
\end{IEEEeqnarray*}
It now follows that when a message $m_x$ is updated to $m'_x$, each inner node $u_{b_0,\ldots,b_j}$ such that $\bin(x)[0{:}j-p-1]=(b_1,\ldots,b_{j-p})$ can be updated from $u_{b_0,\ldots,b_j}$ to $u'_{b_0,\ldots,b_j}$ using the functions $h_{i,j}$ as follows:
\begin{IEEEeqnarray*}{rCl}
    u'_{b_0,\ldots,b_j} &=&  
         h_{i,j,(b_{j-p+1},\ldots,b_j)}(m'_x) + \hspace{-2em}  \sum_{\substack{i \neq x\colon \\ i\colon\bin(i)[0{:}j-p-1]=(b_1,\ldots,b_{j-p})}} \hspace{-3em} h_{i,j,(b_{j-p+1},\ldots,b_j)}(m_i) \\
    &=& u_{b_0,\ldots,b_j} + h_{i,j,(b_{j-p+1},\ldots,b_j)}(m'_x) - h_{i,j,(b_{j-p+1},\ldots,b_j)}(m_x)
\end{IEEEeqnarray*}
When the partial digest functions are homomorphic in their input, the expression above can be written as
\[h_{i,j,(b_{j-p+1},\ldots,b_j)}(m'_x) - h_{i,j,(b_{j-p+1},\ldots,b_j)}(m_x) = h_{i,j,(b_{j-p+1},\ldots,b_j)}(m'_x-m_x).\]
This lets us calculate the updated internal node using only the knowledge of the message diff $m_i'-m_i$.
We provide examples of homomorphic tree constructions in Section~\ref{sec:tree-homomorphic-examples} with homomorphic partial digest functions.
In our general formulation, we assume that the opening proof $\pi_i$ for each message $m_i$ at index $i$ is a function of the inner nodes (and/or their siblings) on the path from the message to the root.

Unlike in Section~\ref{sec:merkle-update}, homomorphic trees enable calculating the new inner nodes after message updates using \emph{only} the new and the old updated messages, in particular using only their difference.
Hence, we can construct a tree that achieves the same complexity for the update information size as algebraic VCs, albeit at the expanse of the proof update complexity, \emph{without requiring the users to keep track of the old messages or to calculate the tree from scratch given all messages} (see Appendix~\ref{sec:why-homomorphic} for further discussion).
This is in contrast to Merkle treess based on SHA256.
The update and proof update algorithms of such a homomorphic tree with no structured update information and the same asymptotic complexity as algebraic VCs is described in Appendix~\ref{sec:homomorphic-no-update-information}.
Since the homomorphic trees can achieve both extremes in terms of update information size and update runtime (Table~\ref{tab:vc-schemes}), with a smart structuring of the update information, they can support sublinear update.
We show how in the next subsection.

\subsection{Structuring the Update Information}
\label{sec:structuring-update-information-ht}

\begin{algorithm}[ht!]
    \captionsetup{font=small} 
    \caption{Algorithms for updating a homomorphic tree. Each user knows the total number of leaves $N$. The recursive algorithm $\textsc{UpdateNode}$, parameterized by $\nu \in [0,1]$ and $p$, takes an index as input, and checks if the new value of the node at that index is to be published as part of the update information $U$. If so, it appends the new value to $U$, and recursively calls itself on the children of the node. Not all of $U$ and $(i, m'_{i}-m_{i})_{i \in [N]}$ are passed to the proof update algorithm and its relevant parts are read selectively to keep the runtime at a minimum.
    Here, $p$ denotes the dependency of the inner node on the messages underlying its parents.
    For instance, $p = 0$ means that the inner nodes depends only on the messages under itself, whereas $p = 1$ means that it depends on the messages under its parent.
    }
    \label{alg.vc.sublinear.p}
    \begin{algorithmic}[1]\footnotesize
    \Alg{\sc Update}{$C, (i, m'_{i}-m_{i})_{i \in [N]}$, \T, $p$}
        \Let{U}{\textsc{Empty}()}
        \Alg{\sc UpdateNode}{$\mathsf{idx}, p$}
            \Let{b_0,\ldots,b_d}{\mathsf{idx}}
            \Let{\mathcal{S}}{\{j \in [k] \colon\bin(i_j)[0:d-p-1] = (b_1, \ldots, b_{d-p})\}}
            \If{$|\mathcal{S}| > k^{1-\nu}$}
                \Let{\T[b_0, \ldots, b_d]}{\T[b_0, \ldots, b_d] \cdot \prod_{j \in \mathcal{S}} h_{i_j,d,(b_{d-p+1},b_d)}(m'_{i_j}-m_{i_j})}
                \Let{U[b_0, b_1, \ldots, b_d]}{\T[b_0, b_1, \ldots, b_d]}
                \State $\textsc{UpdateNode}((b_0,\ldots,b_d,0))$
                \State $\textsc{UpdateNode}((b_0,\ldots,b_d,1))$
            \EndIf
        \EndAlg
        \State $\textsc{UpdateNode}((b_0))$
        \Let{C'}{\mathrm{new}\ \mathrm{VC}} \Comment{Commitment need not be the same as the root.} \label{line:vc-update}
        \State\Return $(C', U, \T)$
    \EndAlg
    \Alg{\sc ProofUpdate}{$C, \pi_x , m'_x, x, U, p$}
        \Let{\pi'_x}{\{\}}
        \Let{(b_1,\ldots,b_h)}{\bin(x)}
        \For{$d = h, \ldots, 1$}
            \If{$(b_0,b_1 \ldots\overline{b}_d) \in U$}
                \Let{\pi'_x[b_0,b_1, \ldots, b_d]}{U[b_0,b_1, \ldots, b_d]}
            \Else
                \Let{\mathcal{S}}{\{j \in [k] \colon\bin(i_j)[0:d-p-1] = (b_1, \ldots, b_{d-p})\}}
                \Let{\pi'_x[b_0,\ldots,b_d]}{\pi_x[b_0,\ldots,b_d] \cdot \prod_{j \in \mathcal{S}} h_{i_j,d, (b_{d-p+1},b_d)}(m'_{i_j}-m_{i_j})}
            \EndIf
        \EndFor
        \State\Return $\pi'_x$
    \EndAlg
    \end{algorithmic}
\end{algorithm}

We next describe the update and proof update algorithms that enable VCs based on homomorphic trees to achieve sublinear complexity as a function of the parameter $\nu$ (Alg.~\ref{alg.vc.sublinear.p}).
These algorithms are parameterized by $p$ that denotes the dependence of the inner nodes on the messages under other nodes.

\subsubsection{Update Information}
When the messages $(i_j, m_{i_j})_{j \in [k]}$ change to $(i_j, m'_{i_j})_{j \in [k]}$, the update information $U$ is generated recursively using the following algorithm:
\begin{enumerate}
    \item Start at the root $u_{b_0}$. 
    Terminate the recursion at an inner node if there are $k^{1-\nu}$ or less updated messages under the $p$-th parent of that node (or under the root if the node is at a depth less than $p$ measured from the root).
    Here, the $0$-th parent refers to the node itself, the $1$-st parent is the node's direct parent, the $2$-nd parent is its parent's parent and so on.
    \item If the recursion is not terminated at a node, then publish the new value of the node as part of $U$, and apply the same algorithm to the left and the right children of the node.
\end{enumerate}
The new values of the inner nodes included in $U$ are listed from the smallest to the largest depth in the canonical left to right order.

\subsubsection{Proof Update}
When the messages $(i_j, m_{i_j})_{j \in [k]}$ are updated to $(i_j, m'_{i_j})_{j \in [k]}$, a user first retrieves the inner nodes within its opening proof that are published as part of the update information.
It then calculates the non-published inner nodes within the proof using the partial digests.
For instance, consider a user with the proof $\pi_{x} = (N, x, u_{b_0}, u_{b_0,b_1}, \ldots, u_{b_0,b_1,\ldots,b_h})$ for some message $m_x$, $(b_1, \ldots, b_h) = \bin(x)$.
To update the proof, the user first replaces the inner nodes whose new values are provided by $U$: 
$u'_{b_0,\ldots,b_d} \xleftarrow[]{} U[b_0, \ldots, b_d]$, $d \in [h]$, if $U[b_0, \ldots, b_d] \neq \bot$. 
Otherwise, the user finds the new values at the nodes $u_{b_0,\ldots,b_d}$, $d \in [h]$, using the functions $h_{x,d,c}$:
\begin{IEEEeqnarray*}{llCl}
& \mathcal{S} &=& \{j \colon \bin(i_j)[0:d-p-1] = (b_1, \ldots, b_{d-p})\} \\
& u'_{b_0, \ldots, b_{d}} &=& u_{b_0, \ldots, b_{d}} \cdot \prod_{j \in \mathcal{S}} h_{i_j,d,(b_{d-p+1},\ldots,b_d)}(m'_{i_j} -m_{i_j}) 
\end{IEEEeqnarray*}

\subsubsection{Complexity}
Finally, we prove bounds on the size and time complexity for proof updates. 
\begin{theorem}
\label{thm:complexity-homomorphic-tree}
Complexity of the update information size and the runtime of proof updates are as follows: $g_1(k) = k^\nu$ and $g_2(k) = k^{1-\nu}$.
\end{theorem}

\begin{proof}
Consider the set $\mathcal{P}$ consisting of the $p$-th parents of the nodes in $U$ such that no child of a node in $\mathcal{P}$ is the $p$-th parent of a node in $U$.
By definition, there are over $k^{1-\nu}$ updated messages within the subtree rooted at the $p$-th parent of each node $u \in U$.
Since there are $k$ updated messages in total, and by definition of $\mathcal{P}$, the subtrees rooted at the nodes in $\mathcal{P}$ do not intersect at any node, there must be less than $k/k^{1-\nu} = k^\nu$ inner nodes in $\mathcal{P}$.
Since the total number of published inner nodes is given by the nodes in $\mathcal{P}$, the nodes on the path from the root to each node in $\mathcal{P}$ and the descendants of the nodes in $\mathcal{P}$ up to the $p$-th descendants, this number is bounded by $2^{p} k^\nu \log{(N)} = \Tilde{\Theta}(k^\nu)$.
Hence,
$|U| = \Theta(2^{p} k^\nu \log{(N)}(\log{(N)}+|G|)) = \Tilde{\Theta}(k^\nu)|G| = \Tilde{\Theta}(k^\nu) \lambda$,
which implies $g_1(k) = k^\nu$.

For each inner node in its opening proof, the user can check if a new value for the node was provided as part of $U$, and replace the node if that is the case, in at most $\Theta(\log{(N)}+|G|)$ time by running a binary search algorithm over $U$.
On the other hand, if the new value of a node in the proof is not given by $U$, the user can calculate the new value by adding the partial digests of $k^{1-\nu}$ message updates to the old value of the node.
This is because there can be at most $k^{1-\nu}$ updated messages within the subtree rooted at the $p$-th parent of an inner node, whose new value was not published as part of $U$.
Assuming that each partial digest can be calculated and added to each inner node on an opening proof in constant time $T_G$, the total time complexity of a proof update becomes at most 
\begin{IEEEeqnarray*}{C}
\Theta(\log{(N)}(\log{(N)}+|G|+k^{1-\nu}T_G)) = \Tilde{\Theta}(k^{1-\nu}) T_G,
\end{IEEEeqnarray*}
which implies $g_2(k) = k^{1-\nu}$.
\end{proof}

\begin{figure}
    \centering
    \includegraphics[width=0.8\linewidth]{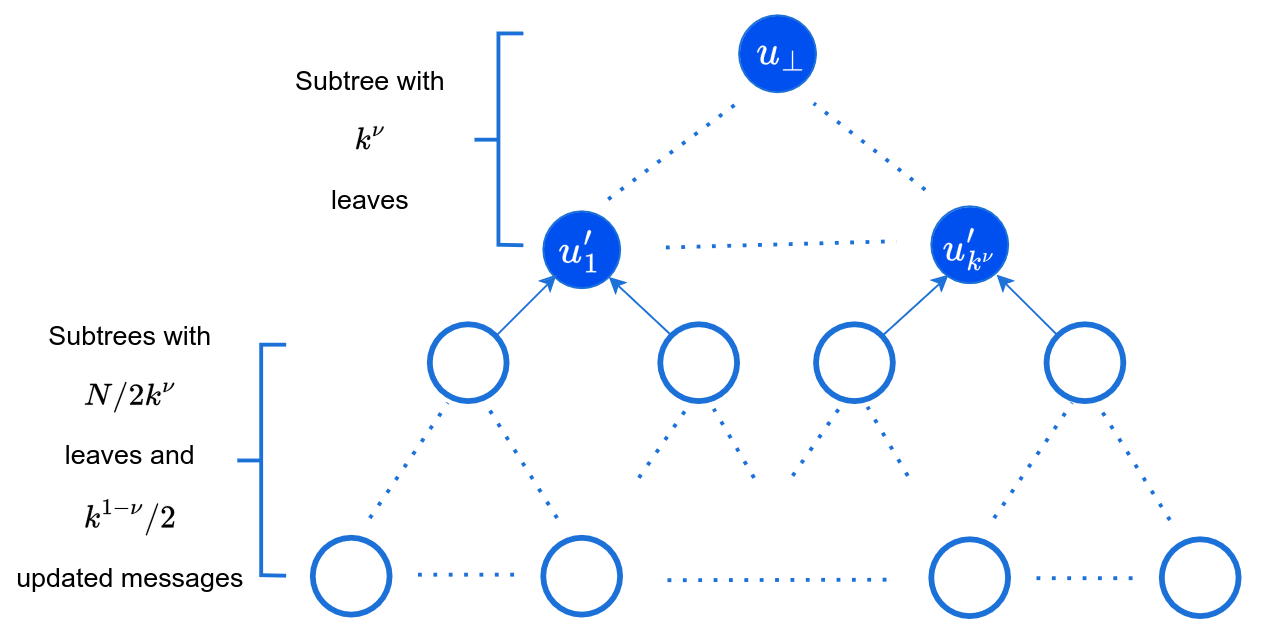}
    \caption{\small Homomorphic tree example for $p = 0$. The new values of the inner nodes with solid blue color are published as part of the updated information.}
    \label{fig:homomorphic-merkle-tree}
\end{figure}

To illustrate the proof above, consider the homomorphic tree in Figure~\ref{fig:homomorphic-merkle-tree} where $k$ messages are updated.
Suppose $p = 0$, \ie, each inner node depends only on the nodes within its subtree.
Suppose there are $k^{1-\nu}/2$ updated messages among the first $N/2k^\nu$ messages $m_0, \ldots, m_{N/2k^\nu-1}$, another $k^{1-\nu}/2$ updated messages among the second $N/2k^\nu$ messages $m_{N/2k^\nu}, \ldots m_{2N/2k^\nu - 1}$ and so on.
In this case, the algorithm identifies the inner nodes within the subtree at the top of the tree (whose nodes are denoted in solid blue) and publishes their new values as part of the update information (when $p = 1$, the algorithm identifies and publishes the \emph{children} of the solid blue nodes).
This is because there are $k^{1-\nu}$ updated messages under each inner node and leaf of this subtree, denoted by $u'_{i}$, $i = 1, \ldots, k^\nu$, whereas under the children of these leaf nodes there are less than $k^{1-\nu}$ updated messages.
Thus, each user can update its opening proof by downloading the new values of the top $\log{k^\nu}$ inner nodes within its proof from the update information. 
There are at most $k^{1-\nu}/2$ updated messages under each of the remaining $\log{N/k^\nu}$ nodes in the proof; hence, the user can find their updated values in $\Theta(\log{(N)}k^{1-\nu})$ time.
Note that in this example, and in general when the updated messages are distributed uniformly among the leaves, the size of the update information becomes $\Theta(k^\nu)\lambda$ rather than $\Theta(k^\nu \log{N}) \lambda$.

\subsection{Constructions for Homomorphic Trees}
\label{sec:tree-homomorphic-examples}

Homomorphic tree constructions are based on Authenticated Multipoint Evaluation Trees (AMTs)~\cite{ACZ+20} and Hyperproofs~\cite{hyperproofs}, both of which use pairings.
We next describe a homomorphic tree construction based on AMTs.
In an AMT, each inner node is a KZG commitment to a quotient polynomial.
Let $\phi(X)$ denote the degree $N-1$ polynomial such that $\phi(i) = m_i$.
We define the range operation $\range(b_0, \ldots, b_j)$ as the set of integers $i$ such that for each $i$, there exists a (potentially empty) binary sequence $b^i_{1}, \ldots, b^i_{h-j}$ that satisfies $\bin(i) = b_1, \ldots, b_j, b^i_{1}, \ldots, b^i_{h-j}$.
Given the $\range(.)$ operation, for each $j \in [h]$, we recursively define the quotient and remainder polynomials $q_{b_0, \ldots, b_j}$ and $r_{b_0, \ldots, b_j}$ as follows:
\begin{IEEEeqnarray*}{llCl}
\text{for the root: } & \phi(X) &=& q_{b_0}(X) \prod_{i \in \range(b_0)} (X-i) + r_{b_0}(X) \\*
\text{for internal nodes: } & r_{b_0, \ldots, b_{j-1}}(X) &=& q_{b_0, \ldots, b_j}(X) \prod_{i \in \range(b_0, \ldots, b_j)} (X-i) + r_{b_0, \ldots, b_j}(X)
\end{IEEEeqnarray*}
For the leaves at indices $i \in [N]$, it holds that $r_{b_0, \bin(i)} = m_i$.

The inner nodes of the AMT are KZG commitments to the quotient polynomials at the corresponding positions:
\begin{IEEEeqnarray*}{C}
u_{b_0, \ldots, b_j} = [q_{b_0, \ldots, b_j}(X)].
\end{IEEEeqnarray*}
The vector commitment to the messages $(m_i)_{i \in N}$ is the KZG commitment $C = [\phi(X)]$ to the polynomial $\phi(X)$ (thus, the VC can be updated at Line~\ref{line:vc-update} of Alg.~\ref{alg.vc.sublinear.p} in the same way as in Section~\ref{sec:kzg-update}).
The opening proof $\pi_i$ for the message $m_i$ at index $i$, $\bin(i) = b_1, \ldots, b_h$, consists of the inner nodes on the path from the message to the root: $\pi_i = (N, i, (u_{b_0, \ldots, b_j})_{j = 0, \ldots, h})$.
To verify an opening proof $(N, i, (u_{b_0, \ldots, b_j})_{j = 0, \ldots, h})$ for a message $m_i$ with respect to the commitment $C$, the verifier checks if the following pairing relation holds:
\begin{IEEEeqnarray*}{C}
e(C, g) = e(g^{m_i}, g) \cdot \prod_{j = 0}^h e\left(u_{b_0, \ldots, b_j}, \left[\prod_{i \in \range(b_0, \ldots, b_j)} (X-i)\right]\right)
\end{IEEEeqnarray*}
(Although the verification above requires the public parameters to include commitments to $O(N)$ polynomials of the form $[\prod_{i \in \range(b_0, \ldots, b_j)} (X-i)]$, \cite{ACZ+20} reduces the size of the public parameters needed to calculate these commitments to $O(\log{(N)})$ by restricting the evaluation points to roots of unity.)

Let $(u^i_{b_0, \ldots, b_j})_{(b_1, \ldots, b_j) \in 2^{\{0,1\}^n}}$ and $(q^i_{b_0, \ldots, b_j}(X))_{(b_1, \ldots, b_j) \in 2^{\{0,1\}^n}}$ denote respectively the inner nodes of the AMT of the Lagrange basis polynomial $L_i(X)$ and the corresponding quotient polynomials at those inner nodes.
By the homomorphism of AMTs, each inner node of an AMT for a polynomial $\phi(X)$ can be expressed as a weighted product of the corresponding inner nodes of the AMTs of the Lagrange basis polynomials $L_i(X)$, $i \in [N]$.
Moreover, we observe that each remainder polynomial $r_{b_0, \ldots, b_j}$ depends only on the messages $m_i$, $i \in \range(b_0, \ldots, b_j)$, \ie, the messages that are under its node.
As $q_i(b_0, \ldots, b_j)(X)$ is the quotient polynomial obtained by dividing $r_{b_0, \ldots, b_{j-1}}(X)$ by $\prod_{i \in \range(b_0, \ldots, b_j)}(X - i)$, it holds that
\begin{IEEEeqnarray*}{llCl}
\text{for the root: } & q_{b_0}(X) &=& \sum_{i = 1}^N m_i q^i_{b_0}(X) \\*
\text{for the root: } & u_{b_0} &=& \prod_{i = 1}^N (u^i_{b_0})^{m_i} \\*
\text{for internal nodes: } & q_{b_0, \ldots, b_j}(X) &=& \sum_{i \in \range(b_0, \ldots, b_{j-1})} m_i q^i_{b_0, \ldots b_j}(X) \\*
\text{for internal nodes: } & u_{b_0, \ldots, b_j} &=& \prod_{i \in \range(b_0, \ldots, b_{j-1})} (u^i_{b_0, \ldots, b_j})^{m_i}
\end{IEEEeqnarray*}
Therefore, AMTs are homomorphic trees, and their partial digest function is characterized by $p=1$:
\[h_{i,j,c}(m_i) = (u^i_{b_0,\bin(i)[0{:}j-2],c})^{m_i}\]

When a message at index $i$, $\bin(i) = (b_1, \ldots, b_h)$, is updated, the new values (denoted by $u'$) of the inner nodes $u_{\bot}$, $u_{0,1}$, $u_{1,0}$ and $u_{b'_0, \ldots, b'_j}$, $j \geq 2$, such that $\bin(i)[0{:}j-2]=(b'_1,\ldots,b'_{j-1})$ (\ie, $b'_\ell=b_\ell$ for $\ell \in [j-1]$),  can be found using the inner nodes of the AMTs of Lagrange basis polynomials as shown below:
\begin{IEEEeqnarray*}{C}
u'_{b'_0, \ldots, b'_j} = u_{b'_0, \ldots, b'_j} \cdot (u^i_{b'_0,\ldots,b'_j})^{m'_i-m_i} = u_{b'_0, \ldots, b'_j} \cdot h_{i, j, (b'_j)}(m'_i-m_i)
\end{IEEEeqnarray*}
Therefore, updating inner nodes requires access to the AMTs of the Lagrange basis polynomials $L_i(X)$, $i \in [N]$, which must thus be published as public parameters for updateability.

\subsection{A Concrete Evaluation}
\label{sec:eval-homomorphic-tree}
We next estimate the size of the update information and the proof update time after observing an Ethereum block with ERC20 token transfers.
Suppose the block has the target size of $15$ million gas~\cite{gasandfees}, and each token transfer updates the balance of two distinct accounts stored at separate leaves of the homomorphic tree.
Since each ERC20 token transfer consumes approximately $65,000$ gas, there are $\sim 230$ such transactions in the block, and the block updates $k = 460$ accounts.

We consider a binary AMT with degree $2$ that commits to $N = 256^3 = 2^{24}$ accounts.
For comparison, $256^3 \approx 16.7$ million, matching in magnitude the total number of cumulative unique Ethereum addresses, which is $200$ million as of $2023$~\cite{eth-addresses}.
Thus, the AMT has height $\log_2{(N)} = h = 24$, and each opening proof consists of $25$ KZG commitments, $u_{b_0}, \ldots, u_{b_0, \ldots, b_h}$, which are elements of the pairing-friendly elliptic curve $\mathrm{BLS}12\_381$ and have size $|G| = 48$ bytes~\cite{verkle}.
This implies a proof size of $(\log_2{(N)}+1)|G| = 1.20 $ kBytes.

When $\nu = 1/2$ and $460$ accounts are updated, in the worst case, the update information consists of $2 \lceil k^\nu \rceil \log{(N)} = 1030$ inner nodes.
This implies an update information size of $|U| = 2 \lceil k^\nu \rceil \log{(N)} (\log{(N)}+|G|) = 52.53$ kBytes.
This is comparable to the size of the Ethereum blocks, which are typically below $125$ kBytes~\cite{block-size}.

As for the update time, in the worst case, each user has to add the partial digests of up to $\lceil k^{1-\nu} \rceil = 22$ updated messages to each inner node appearing within its opening proof, one from each height of the homomorphic tree.
Following the benchmarks published in \cite{blst-benchmark} for the specified curve, these operations can take up to $k^{1-\nu} \log_2{(N)}\ 0.000665471\text{ s} = 0.34$ seconds on commodity hardware, given a runtime of $665471$ nanoseconds per exponentiation of a group element with a message value.
This is again comparable to the $12$ second inter-arrival time of Ethereum blocks.

Table~\ref{tab:comparison-homomorphic-tree} compares the number of published inner nodes $2 \lceil k^\nu \rceil \log{(N)}$, the total update information size $2 \lceil k^\nu \rceil \log_2{(N)} ) (\log_2{(N)} + |G|)$, the number of group exponentiations per proof update $k^{1-\nu} \log_2{(N)}$ and the proof update time $k^{1-\nu} \log_2{(N)}\ 0.000665471\text{ s}$ at $\nu = 0, 0.25, 0.5, 0.75, 1$.
The degree of the homomorphic Merkle tree and the opening proof size are fixed at $2$ and $25$ inner nodes ($|\pi| = 1.20$ kBytes) respectively.
\begin{table}[]
\centering
\begin{tabular}{|c|c|c|c|c|}
 \hline
  $\nu$ & \# inner nodes & $|U|$ (kBytes) & $\#$ group exp. & Time (s) \\  [0.5ex]
 \hhline{|=|=|=|=|=|}
  $0$ & $1$ & $0.048$ & $11040$ & $7.35$ \\
  \hline
  $0.25$ & $223$ & $11.37$ & $2384$ & $1.59$ \\
  \hline
  $0.50$ & $1030$ & $52.53$ & $515$ & $0.34$ \\
  \hline
  $0.75$ & $4768$ & $243.17$ & $112$ & $0.075$ \\
  \hline
  $1$ & $22080$ & $1126.08$ & $0$ & $\sim0$ \\
 \hline
\end{tabular}
\vspace{0.5cm}
\caption{For different trade-off points between the update information size and proof update complexity, parameterized by $\nu$, the table shows the (worst case) number of published inner nodes $2 \lceil k^\nu \rceil \log{(N)}$, the total update information size $2 \lceil k^\nu \rceil \log{(N)} (\log{(N)} + |G|)$, the number of group exponentiations per proof update $k^{1-\nu} \log{(N)}$ and the estimated time for a single proof update $k^{1-\nu} \log{(N)}\ 0.000665471\text{ s}$.
There are $N = 2^{24}$ accounts in total, $k=460$ updates at the accounts, the inner nodes have size $|G|=48$ Bytes, and each group exponentiation takes $T_G = 0.000665471\text{ s}$.}
\label{tab:comparison-homomorphic-tree}
\vspace{-0.5cm}
\end{table}

We also calculate the size of the public parameters needed for the update operation.
For this purpose, we find the total size of the AMTs of the $N$ Langrange basis polynomials.
Each AMT of a Lagrange basis polynomial contains at most $2\log{(N)}+1$ inner nodes that are not equal to the identity group element.
Hence, the total size of the public parameters becomes $(2N \log{(N)} + N)|G|$, which is equal to $36.46$ GBytes.

We compare the numbers from the concrete evaluation of homomorphic trees to those for Verkle trees in Section~\ref{sec:evaluation}.

\section{Post-Quantum Secure Vector Commitments with Sublinear Update}
\label{sec:sublinear-complexity}

In this section, we describe a different family of VCs with sublinear update, parameterized by the values $\nu \in (0,1)$ and characterized by the functions $(g_1,g_2) = (k^{\nu}, k^{1-\nu})$.
They are based on homomorphic Merkle trees.
These VCs do not require any trusted setup and are secure against quantum computers.
However, they are less performant compared to the homomorphic tree constructions of Section~\ref{sec:sublinear-complexity-efficient}.

\subsection{Homomorphic Merkle Trees}
\label{sec:merkle-tree-homomorphic}

We first introduce homomorphic Merkle trees where messages placed in the leaves take values in a set $\mathcal{M}$. 
We will use two collision-resistant hash functions $\tilde{f} \colon \mathcal{D}\times\mathcal{D} \to \mathcal{R}$ and $f \colon \mathcal{M} \to \mathcal{R}$,
where both $\mathcal{M}$ and $\mathcal{D}$ are vector spaces over some field $\mathbb{F}$, 
and $\mathcal{R}$ is an arbitrary finite set.
We will also need an injective  mapping $g: \mathcal{R} \to \mathcal{D}$, which need not be efficiently computable.
We use $g^{-1}: \mathcal{D} \to \mathcal{R}$ to denote the inverse of~$g$,
meaning that $g^{-1}(g(x)) = x$ for all $x \in \mathcal{R}$.
We require that $g^{-1}$ be efficiently computable.

Now, for $j \in [h]$, where $h$ is the height of the tree, every node $u_{b_0,\ldots,b_j} \in \mathcal{D}$ of the homomorphic Merkle tree is characterized by the following expressions:
\begin{IEEEeqnarray*}{llCl}
\text{a leaf node:} & g^{-1}(u_{b_0,\bin(i)}) &=& f(m_i) \\*    %
\text{an internal node:} \quad & g^{-1}(u_{b_0,\ldots,b_j}) &=& \tilde{f}(u_{b_0,\ldots,b_j,0},\  u_{b_0,\ldots,b_j,1}) \text{ for } j < h
\end{IEEEeqnarray*}

The homomorphic property of the Merkle tree refers to the fact that 
there are efficiently computable functions
\[
   h_{i,j}: \mathcal{D} \to \mathcal{D}  \qquad \text{for $i \in [N]$ and $j \in [h]$,}
\]
such that every inner node $u_{b_0,\ldots,b_j} \in \mathcal{D}$ can be expressed as
\begin{IEEEeqnarray*}{rCl}
u_{b_0} &=& \hspace{2.5em} \sum_{i \in [N]} \hspace{0.5em} h_{i,0}(m_i) \\
u_{b_0,\ldots,b_j} &=& \hspace{-1em} \sum_{i\colon\bin(i)[0{:}j-1]=(b_1,\ldots,b_j)} \hspace{-3em} h_{i,j}(m_i).
\end{IEEEeqnarray*}
We refer to the function $h_{i,j}$ as a \emph{partial digest function} 
and refer to $h_{i,j}(m_i)$ as the \emph{partial digest} of $m_i$.
In a homomorphic Merkle tree, every internal node is the sum of the partial digests of the leaves under that node.
We will show in Section~\ref{sec:linear-hash-functions} that each function $h_{i,j}$ can be expressed
as an iterated composition of the functions $f$ and $\tilde{f}$.
Evaluating $h_{i,j}$ requires evaluating the functions $f$ and $\tilde{f}$ exactly $h-j$ times.
In the notation of Section~\ref{sec:tree-homomorphic}, homomorphic Merkle trees correspond to a homomorphic tree with $p = 0$.
Since when $p = 0$, $h_{i,j,c}(m_i) = h_{i, j, \emptyset}(m_i)$ for all $i,j \in [N] \times [h]$ in the notation of Section~\ref{sec:merkle-tree-homomorphic}, we omit the index $c$ and denote the partial digest functions above by $h_{i,j}(m_i)$.

Opening proof for a message consists of \emph{both} children of the internal nodes on the path from the message to the root (as opposed to Merkle opening proofs that contain only the siblings of the internal nodes on the path).
For instance, the opening proof for the message $m_i$ at leaf index $i$, with $\bin(i) = (b_1,\ldots,b_h)$, is $(N, i, (u_{b_0,\ldots,b_j,0},u_{b_0,\ldots,b_j,1})_{j=0,\ldots,h-1})$.
Opening proofs are verified using the functions $f$ and ${\tilde f}$ (not by using the functions $h_{i,j}$).
To verify an opening proof $(N, i, (u_{b_0,\ldots,b_j,0},u_{b_0,\ldots,b_j,1})_{j=0,\ldots,h-1})$ for a message $m_i$ with respect to the root $u_{b_0}$, the verifier checks if the following equalities hold:
\begin{IEEEeqnarray*}{llCl}
\text{for the leaf:} & g^{-1}(u_{b_0,\bin(i)}) &=& f(m_i) \\*    %
\text{for the internal nodes:} \ & g^{-1}(u_{b_0,\ldots,b_j}) &=& \tilde{f}(u_{b_0,\ldots,b_j,0},\  u_{b_0,\ldots,b_j,1}) \text{ for } j = h-1, \ldots, 0.
\end{IEEEeqnarray*}
If so, it accepts the proof, and otherwise it outputs reject.

As in Section~\ref{sec:tree-homomorphic}, consider a homomorphic Merkle tree that commits to four messsages $m_0,m_1,m_2,m_3$.
Since $p = 0$ for homomorphic Merkle trees, its root $u_{\bot}$ and inner nodes $u_{\bot,0}$, $u_{\bot,1}$, $u_{\bot,0,0}$, $u_{\bot,0,1}$, $u_{\bot,1,0}$, $u_{\bot,1,1}$ can be calculated as follows:
\begin{IEEEeqnarray*}{rClrCl}
u_{\bot} &=& h_{0,0}(m_0) + h_{1,0}(m_1) + h_{2,0}(m_2) + h_{3,0}(m_3)\ ; \qquad   &    u_{\bot,0,0} &=& h_{0,2}(m_0) \\*
u_{\bot,0} &=& h_{0,1}(m_0) + h_{1,1}(m_1)\ ;                                      &    u_{\bot,0,1} &=& h_{1,2}(m_1) \\*
u_{\bot,1} &=& h_{2,1}(m_2) + h_{3,1}(m_3)\ ;                                      &    u_{\bot,1,0} &=& h_{2,2}(m_2) \\*
&&                                                                              &    u_{\bot,1,1} &=& h_{3,2}(m_3)
\end{IEEEeqnarray*}
The opening proof for $m_3$ is given by $(4, 3, ((u_{\bot,0}, u_{\bot,1}), (u_{\bot,1,0}, u_{\bot,1,1})))$, and verified by checking the following equations:
\begin{IEEEeqnarray*}{llCl}
\text{for $u_{\bot,1,1}$:} \quad   & g^{-1}(u_{\bot,1,1}) &=& f(m_i) \\*
\text{for $u_{\bot,1}$:}           & g^{-1}(u_{\bot,1}) &=& \tilde{f}(u_{\bot,1,0},\  u_{\bot,1,1}) \\*
\text{for $u_{\bot}$:}             & g^{-1}(u_{\bot}) &=& \tilde{f}(u_{\bot,0},\  u_{\bot,1})
\end{IEEEeqnarray*}

It now follows that 
when a message $m_i$ is updated to $m'_i$, each inner node on the path from the leaf to the root can be updated from $u_{b_0,\ldots,b_j}$ to $u'_{b_0,\ldots,b_j}$ using the functions $h_{i,j}$ as follows:
\[
    u'_{b_0,\ldots,b_j} \ \ =\ \  
         h_{i,j}(m'_i) + \hspace{-2em}  \sum_{\substack{x \neq i\colon \\ \bin(x)[0{:}j-1]= (b_1,\ldots,b_j)}} \hspace{-3em} h_{x,j}(m_x)  \ \ = \ \ 
         u_{b_0,\ldots,b_j} + h_{i,j}(m'_i) - h_{i,j}(m_i)
\]
When the partial digest functions are linear in their input, the expression $h_{i,j}(m'_i) - h_{i,j}(m_i)$ can be written as $h_{i,j}(m'_i) - h_{i,j}(m_i) = \text{sign}(m'_i-m_i)h_{i,j}(|m'_i-m_i|)$.
This lets us calculate the updated internal node using only the knowledge of the message diff $m_i'-m_i$.
We provide examples of homomorphic Merkle tree constructions in Section~\ref{sec:linear-hash-functions} with linear partial digest functions $h_{i,j}$.
Homomorphic Merkle proofs in these constructions consist of the two siblings of the inner nodes on the path from the proven message to the root and the vector commitment itself is given by $g^{-1}(b_\bot)$ (Section~\ref{sec:linear-hash-functions}).

\subsection{Structuring the Update Information}
\label{sec:structuring-update-information}

\begin{algorithm}[ht!]
    \captionsetup{font=small} 
    \caption{Algorithms for a homomorphic Merkle tree. Each user knows the total number of leaves $N$. The recursive algorithm $\textsc{UpdateNode}$, parameterized by $\nu \in [0,1]$, takes an index as input, and checks if the new value of the node at that index is to be published as part of the update information $U$. If so, it appends the new value to $U$, and recursively calls itself on the children of the node. Not all of $U$ and $(i, m'_{i}-m_{i})_{i \in [N]}$ are passed to the proof update algorithm and its relevant parts are read selectively to keep the runtime at a minimum.
    }
    \label{alg.vc.sublinear}
    \begin{algorithmic}[1]\footnotesize
    \Alg{\sc Update}{$C, (i, m'_{i}-m_{i})_{i \in [N]}$, \T}
        \Let{U}{\textsc{Empty}()}
        \Alg{\sc UpdateNode}{$\mathsf{idx}$}
            \Let{b_0,\ldots,b_d}{\mathsf{idx}}
            \Let{\mathcal{S}}{\{j \in [k] \colon\bin(i_j)[0:d-1] = (b_1, \ldots, b_d)\}}
            \If{$|\mathcal{S}| > k^{1-\nu}$}
                \Let{\T[b_0, \ldots, b_d]}{\T[b_0, \ldots, b_d] + \sum_{j \in \mathcal{S}} \text{sign}(m'_{i_j}-m_{i_j})h_{i_j,d}(|m'_{i_j} -m_{i_j}|)}
                \Let{U[b_0, b_1, \ldots, b_d]}{\T[b_0, b_1, \ldots, b_d]}
                \State $\textsc{UpdateNode}((b_0,\ldots,b_d,0))$
                \State $\textsc{UpdateNode}((b_0,\ldots,b_d,1))$
            \EndIf
        \EndAlg
        \State $\textsc{UpdateNode}((b_0))$
        \Let{C'}{\T[b_0]}
        \State\Return $(C', U, \T)$
    \EndAlg
    \Alg{\sc ProofUpdate}{$C, \pi_x , m'_x, x, U$}
        \Let{\pi'_x}{\{\}}
        \Let{(b_1,\ldots,b_h)}{\bin(x)}
        \For{$d= h, \ldots, 1$}
            \If{$(b_0,b_1 \ldots\overline{b}_d) \in U$}
                \Let{\pi'_x[b_0,b_1 \ldots\overline{b}_d]}{U[b_0,b_1 \ldots\overline{b}_d]}
            \Else
                \Let{\mathcal{S}}{\{j \in [k] \colon\bin(i_j)[0:d-1] = (b_1, \ldots, \overline{b}_d)\}}
                \Let{\pi'_x[b_0,\ldots,\overline{b}_d]}{\pi_x[b_0,\ldots,\overline{b}_d] + \sum_{j \in \mathcal{S}} \text{sign}(m'_{i_j}-m_{i_j})h_{i_j,d}(|m'_{i_j} -m_{i_j}|)}
            \EndIf
        \EndFor
        \State\Return $\pi'_x$
    \EndAlg
    \end{algorithmic}
\end{algorithm}

Update and proof update algorithms that enable homomorphic Merkle trees to achieve sublinear complexity are described by Alg.~\ref{alg.vc.sublinear}.
Update information is generated using the same algorithm as in Section~\ref{sec:structuring-update-information-ht} for $p = 0$.

When the messages $(i_j, m_{i_j})_{j \in [k]}$ are updated to $(i_j, m'_{i_j})_{j \in [k]}$, a user first retrieves the inner nodes within its Merkle proof that are published as part of the update information.
It then calculates the non-published inner nodes within the proof using the partial digests.
For instance, consider a user with the proof $(u_{b_0,\overline{b}_1}, u_{b_0,b_1,\overline{b}_2}, \ldots, u_{b_0,b_1,b_2,\ldots,\overline{b}_h})$ for some message $m_x$, $(b_1, \ldots, b_h) = \bin(x)$.
To update the proof, the user first checks the update information $U$ and replaces the inner nodes whose new values are provided by $U$: 
$u'_{b_0,b_1,\ldots,\overline{b}_d} \xleftarrow[]{} U[b_0,b_1 \ldots\overline{b}_d]$, $d \in [h]$, if $U[b_0,b_1 \ldots\overline{b}_d] \neq \bot$. 
Otherwise, the user finds the new values at the nodes $u_{b_0,b_1,\ldots,\overline{b}_d}$, $d \in [h]$, using the functions $h_{x,d}$:
\begin{IEEEeqnarray*}{rCl}
u'_{b_0, \ldots, b_{d-1},\overline{b}_{d}} &=& u_{b_0, \ldots, b_{d-1},\overline{b}_{d}} \\
&&\ + \sum_{j \in [k]} 1_{\bin(i_j)[:d] = (b_1, \ldots,\overline{b}_{d})} \left(\text{sign}(m'_{i_j}-m_{i_j})h_{i_j,d}(|m'_{i_j} -m_{i_j}|)\right) 
\end{IEEEeqnarray*}

Theorem~\ref{thm:complexity-homomorphic-tree} also holds for homomorphic Merkle trees:

\begin{theorem}
\label{thm:complexity-hmt}
Complexity of the update information size and the runtime of proof updates are as follows: $g_1(k) = k^\nu$ and $g_2(k) = k^{1-\nu}$.
\end{theorem}

\begin{proof}
Let $\mathcal{U}$ denote the subset of the inner nodes published by the algorithm as part of $U$ such that no child of a node $u \in \mathcal{U}$ is published.
Then, there must be over $k^{1-\nu}$ updated messages within the subtree rooted at each node $u \in \mathcal{U}$.
Since there are $k$ updated messages, and by definition of $\mathcal{U}$, the subtrees rooted at the nodes in $\mathcal{U}$ do not intersect at any node, there must be less than $k/k^{1-\nu} = k^\nu$ inner nodes in $\mathcal{U}$.
Since the total number of published inner nodes is given by $\mathcal{U}$ and the nodes on the path from the root to each node $u \in \mathcal{U}$, this number is bounded by $k^\nu \log{(N)} = \Tilde{\Theta}(k^\nu)$.
Hence,
$|U| = \Theta(k^\nu \log{(N)}(\log{(N)}+|H|)) = \Tilde{\Theta}(k^\nu)|H| = \Tilde{\Theta}(k^\nu) \lambda$,
which implies $g_1(k) = k^\nu$.

For each inner node in its Merkle proof, the user can check if a new value for the node was provided as part of $U$, and replace the node if that is the case, in at most $\Theta(\log{(N)}+|H|)$ time by running a binary search algorithm over $U$.
On the other hand, if the new value of a node in the proof is not given by $U$, the user can calculate the new value after at most $k^{1-\nu}\log{(N)}$ function evaluations.
This is because there can be at most $k^{1-\nu}$ updated messages within the subtree rooted at an inner node, whose new value was not published as part of $U$.
This makes the total time complexity of a proof update at most 
\begin{IEEEeqnarray*}{C}
\Theta(\log{(N)}(\log{(N)}+|H|+k^{1-\nu}\log{(N)}T_f)) = \Tilde{\Theta}(k^{1-\nu}) T_f,
\end{IEEEeqnarray*}
which implies $g_2(k) = k^{1-\nu}$.
\end{proof}

\subsection{Constructions for Homomorphic Merkle Trees}
\label{sec:linear-hash-functions}

Homomorphic Merkle trees were proposed by \cite{PT11,PSTY13,QZC+14}.
They use lattice-based hash functions, and their collision-resistance is proven by reduction to the hardness of the gap version of the shortest vector problem ($\mathsf{GAPSVP}_\gamma$), which itself follows from the hardness of the small integer solution problem.
We next describe the construction introduced by \cite{PT11}, which is similar to those proposed by later works \cite{PSTY13,QZC+14} (an alternative construction is provided in Appendix~\ref{sec:appendix-alternative}).
Its correctness and security follow from \cite[Theorem 4]{PT11}.

Let $L(\mathbf{M})$ denote the lattice defined by the basis vectors $\mathbf{M} \subset \mathbb{Z}^{k \times m}_q$ for appropriately selected parameters $k,m,q$, where $m = 2 k \log q$.
Consider vectors $u \in \{0, \ldots, t\}^{k \log q}$, where $t$ is a small integer.
The (homomorphic) hash functions $f \colon \mathbb{Z}^{k \log q} \to L(\mathbf{M})$ and $\tilde{f} \colon \mathbb{Z}^{k \log q} \times \mathbb{Z}^{k \log q} \to L(\mathbf{M})$ used by \cite{PT11} are defined as $f(x) = \mathbf{M}x$ and $\tilde{f}(x,y) = \mathbf{M} \mathbf{U} x + \mathbf{M} \mathbf{D} y$ respectively.
Here, $\mathbf{U}$ and $\mathbf{D}$ are special matrices that double the dimension of the multiplied vector and shift it up or down respectively.
The remaining entries are set to zero.
For convenience, we define $\mathbf{L} = \mathbf{M} \mathbf{U}$ and $\mathbf{R} = \mathbf{M} \mathbf{D}$.

Since the domain and range of the hash functions are different, to ensure the Merkle tree's homomorphism, authors define a special mapping $g \colon \mathbb{Z}^k_q \to \mathbb{Z}^{k \log{q}}_q$ from the range of the hash functions to their domain.
Here, $g(.)$ takes a vector $\mathbf{v} \in \mathbb{Z}_q$ as input and outputs \emph{a} radix-2 representation for $\mathbf{v}$.
However, as there can be many radix-2 representations of a vector, to help choose a representation that yields itself to homomorphism, authors prove the following result: for any $\textbf{x}_1, \textbf{x}_2, \ldots, \textbf{x}_t \in \mathbb{Z}_q$, there exists \emph{a short} radix-2 representation $g(.)$ such that $g(\textbf{x}_1 + \textbf{x}_2 + \ldots + \textbf{x}_t \mod q) = b(\textbf{x}_1) + b(\textbf{x}_2) + \ldots + b(\textbf{x}_t) \mod q \in \{0, \ldots, t\}^{k \log q}$, where the function $b \colon \mathbb{Z}^k_q \to \{0,1\}^{k\log{q}}$ returns the binary representation of the input vector.
This equality enables the mapping $g(.)$ to \emph{preserve} the hash functions' original homomorphic property.
Then, given an inner node $u_{b_0,\ldots,b_j}$ as input, the homomorphic Merkle tree uses the short radix-2 representation $g(.)$ that enforces the following equality: $u_{b_0,\ldots,b_j} = g(\mathbf{L} u_{b_0,\ldots,b_j,0} + \mathbf{R} u_{b_0,\ldots,b_j,1} \mod q) = b(\mathbf{L} u_{b_0,\ldots,b_j,0}) + b(\mathbf{R} u_{b_0,\ldots,b_j,1}) \mod q$.
Finally, this enables calculating the value of each inner node as a sum of the partial digests $h_{i,j}(.)$ of the messages $m_i$ under the node $u_{b_0,\ldots,b_j}$ (\ie, $(m_i)_{\bin(i)[0{:}j] = (b_0,\ldots,b_j)}$) as outlined in Section~\ref{sec:merkle-tree-homomorphic}, i.e., 
\begin{IEEEeqnarray*}{C}
u_{b_0,\ldots,b_j} = \sum_{i\colon\bin(i)[0{:}j-1]=(b_1,\ldots,b_j)} h_{i,j}(m_i),
\end{IEEEeqnarray*}
where $h_{i,j}(.)$ is expressed in terms of the bits $\bin(i)[j{:}h-1] = (b'_1, \ldots, b'_{h-j})$:
\begin{IEEEeqnarray*}{C}
h_{i,j}(m_i) = f_{b'_1}(f_{b'_2}(\ldots f_{b'_{h-j}}(b(f(m_i)))))
\end{IEEEeqnarray*}
Here, $f_0(.)$ and $f_1(.)$ are defined as $b(\mathbf{L}.)$ and $b(\mathbf{R}.)$ respectively.
Since $b(.)$, binary expansion, is a linear operation and matrix multiplication is linear, $h_{i,j}(.)$ is linear in its input.

Opening proof of a message $m$ consists of its index and $\alpha_i$ and $\beta_i$, $i \in [h]$, $h = \log{(N)}$, where $\alpha_i$ and $\beta_i$ are the children of the inner nodes on the path from $m$ to the root.
The proof can be verified in $\log{(N)}$ time by iteratively checking if $f(m) = g^{-1}(\alpha_h)$ (or $ = g^{-1}(\beta_h)$) and $\tilde{f}(\alpha_i,\beta_i) = g^{-1}(\alpha_{i-1})$ (or $=g^{-1}(\beta_{i-1})$ depending on the message index), where $g^{-1}$ returns a number given its radix-2 representation \cite{PT11}.

Note that both $f$ and $\tilde{f}$ are homomorphic hash functions~\cite{BGG94}.
Other examples of homomorphic hash functions include Pedersen hashes and KZG commitments. 
However, the homomorphic property of the hash function is not sufficient for constructing a homomorphic Merkle tree when the function is combined with the output of other functions in a serial manner as in Merkle trees.
For the lattice-based function, this was possible because of repeated linearity~\cite{PT11}, which refers to the existence of a linear mapping $g(.)$ from the range to the domain of the hash function.
This mapping enabled the iterative hashing within the Merkle tree to preserve the linearity of the hash function.
Such repeated linearity does not exist for Pedersen hashes and KZG commitments as a linear mapping from the range to the domain would imply the violation of the discrete log assumption.
That is why Verkle trees based on KZG commitments are not homomorphic and cannot support sublinear update.

\subsection{A Concrete Evaluation}
\label{sec:eval-homomorphic}

We now estimate the size of the update information and the number of operations to update an opening proof after observing an Ethereum block consisting of ERC20 token transfers.
Suppose the Ethereum state is persisted using the homomorphic Merkle tree construction built with the more performant hash function in \cite[Section 4]{QZC+14}.
As in Section~\ref{sec:eval-homomorphic-tree}, suppose the block has the target size of $15$ million gas~\cite{gasandfees}, and each token transfer updates the balance of two distinct accounts stored at separate leaves of the Verkle tree.
Then, there are $\sim 230$ such transactions in the block, and the block updates $k = 460$ accounts.
We assume that the homomorphic Merkle tree has degree $2$ and commits to $256^3$ accounts as in Section~\ref{sec:eval-homomorphic-tree}.
Each opening proof consists of $2\log{(N)} = 48$ inner nodes.

When $\nu = 1/2$ and $460$ accounts are updated, in the worst case, the update information consists of $\lceil \sqrt{k} \rceil \log{(N)} = 528$ inner nodes.
Our evaluation takes the size of an inner node as given in \cite{QZC+14} for the more performant hash function, namely $|H| = 0.21$ MB (which is equal to the key size in \cite{QZC+14}).
This implies an update information size of $|U| = 110.88$ MBytes and an opening proof size of $|\pi| = 10.08$ MBytes.

As for the update time, in the worst case, each user has to calculate the partial digests of $44$ updated messages at each height of the homomorphic Merkle tree, \ie, the effect of these updated messages on each inner node of its opening proof.
Calculating the partial digest of a message at height $h$ measured from the leaves requires $h$ evaluations of the hash function.
This implies a proof update complexity of $2 \sum_{i=0}^{\log{N}-1} i \min(\lceil \sqrt{k} \rceil, 2^i) = 11,900$ hash evaluations.
To find numerical upper bounds for the update time, we use the hash function evaluation time $T_f = 2.74$ ms, published by \cite{QZC+14} for the more performant function (this is for commodity hardware; \cf \cite{QZC+14} for the details).
This gives an upper bound of $32.6$ seconds for the update time.

Based on the benchmarks in \cite{QZC+14}, Table~\ref{tab:comparison-homomorphic} compares the number of published inner nodes $\lceil k^\nu \rceil \log{(N)}$, the total update information size $\lceil k^\nu \rceil \log{(N)} |H|$ (assuming that the size of each inner node is $|H|$ upper bounded by $0.21$ MBytes), the number of hash evaluations per proof update $2 \sum_{i=0}^{\log{N}-1} i \min(\lceil k^{1-\nu} \rceil, 2^i)$ and the proof update time $2 \sum_{i=0}^{\log{N}-1} i \min(\lceil k^{1-\nu} \rceil, 2^i) T_f$ (assuming that each hash evaluation takes less than $T_f = 2.74$ ms) at $\nu = 0, 0.25, 0.5, 0.75, 1$.
The degree of the homomorphic Merkle tree and the opening proof size are fixed at $2$ and $48$ inner nodes ($|\pi| = 10.08$ MBytes) respectively.
Homomorphic Merkle trees do not require any public parameters.

We compare the numbers from the concrete evaluation of homomorphic Merkle trees to those for Verkle trees in Section~\ref{sec:evaluation}.
\begin{table}[]
\centering
\begin{tabular}{|c|c|c|c|c|}
 \hline
  $\nu$ & \# inner nodes & $|U|$ (MBytes) & $\#$ hash evaluations & Time (s) \\  [0.5ex]
 \hhline{|=|=|=|=|=|}
  $0$ & $1$ & $0.21$ & $227,972$ & $624.6$ \\
  \hline
  $0.25$ & $120$ & $25.20$ & $52,284$ & $143.3$ \\
  \hline
  $0.50$ & $528$ & $110.88$ & $11,900$ & $32.6$ \\
  \hline
  $0.75$ & $2400$ & $504.00$ & $2750$ & $7.54$ \\
  \hline
  $1$ & $11040$ & $2318.40$ & $0$ & $\sim0$ \\
 \hline
\end{tabular}
\vspace{0.5cm}
\caption{For different trade-off points between the update information size and proof update complexity, parameterized by $\nu$, the table shows the number of published inner nodes $\lceil k^\nu \rceil \log{(N)}$, the total update information size $\lceil k^\nu \rceil \log{(N)} |H|$, the number of hash function evaluations per proof update $2 \sum_{i=0}^{\log{N}-1} i \min(\lceil k^{1-\nu} \rceil, 2^i)$ and the proof update time $2 \sum_{i=0}^{\log{(N)}-1} i \min(\lceil k^{1-\nu} \rceil, 2^i) T_f$.
There are $N = 2^{24}$ accounts in total, $k=460$ updates at the accounts, the inner nodes have size $|H|=0.21$ Mbytes, and a hash function evaluation takes $T_f = 2.74$ ms.}
\label{tab:comparison-homomorphic}
\vspace{-0.5cm}
\end{table}

\section{Updating Verkle Trees and Opening Proofs}
\label{sec:verkle-update}

\begin{algorithm}[ht!]
    \captionsetup{font=small} 
    \caption{Update algorithm for a Verkle tree. It takes message updates as input, updates the Verkle tree, \ie the KZG commitments therein, and returns the new Verkle root and the update information. The data structure $\T$ holds the KZG commitments at the inner nodes of the Verkle tree indexed by their positions.
    }
    \label{alg.verkle.update}
    \begin{algorithmic}[1]\footnotesize
    \Alg{\sc Update}{$C, (i, m_{i})_{i \in [N]}, (i, m'_i)_{i \in [N]}, \T$}
        \Let{U}{\textsc{Empty}()}
        \Let{\mathcal{T}'}{\mathcal{T}}
        \For{$j \in [k]$}
            \Let{(b_1,\ldots,b_h)}{\bin(i_j)}
            \Let{\T'[b_0,\ldots,b_h]}{m'_{i_j}}
            \For{$\ell = h-1, \ldots, 0$} 
                \Let{\T'[b_0,\ldots,b_\ell]}{\T[b_0,\ldots,b_\ell][L_{b_{\ell+1}}]^{H(\T'[b_0,\ldots,b_{\ell+1}])-H(\T[b_0,\ldots,b_{\ell+1}])}}
                \Let{U[b_0,\ldots,b_\ell]}{\T'[b_0,\ldots,b_\ell]}
            \EndFor
        \EndFor
        \Let{C'}{h(\T[b_0])}
        \Let{\mathcal{T}}{\mathcal{T}'}
        \State\Return $(C', U, \T)$
    \EndAlg
    \end{algorithmic}
\end{algorithm}
\begin{algorithm}[ht!]
    \captionsetup{font=small} 
    \caption{Proof update algorithm for a Verkle tree. It takes the old Verkle proof $\pi_x$, the KZG commitments $\mathcal{C}$ at the children of the inner nodes from the root to the message, the relevant opening proofs $\Pi'$ held by the user in memory and the update information. It then updates the Verkle proof $\pi_x$ to $\pi'_x$. 
    Here, $\mathcal{C}$ and $\Pi$ are organized as dictionaries indexed by the positions of the maintained KZG commitments and KZG opening proofs.}
    \label{alg.verkle.proofupdate}
    \begin{algorithmic}[1]\footnotesize
    \Alg{\sc ProofUpdate}{$\pi_x , m'_x, x, \Pi, \mathcal{C}, U$}
        \Let{(b_1,\ldots,b_h)}{\bin(x)}
        \Let{-,(G',\pi')}{\pi_x}
        \Let{\Pi'}{\Pi}
        \Let{\mathcal{C}'}{\mathcal{C}}
        \For{$d = h-1, \ldots, 0$}
            \For{$i = 0, \ldots, c-1$}
                \If{$(b_0, \ldots, b_d, i) \in U$}
                    \Let{\mathcal{C}'[b_0,\ldots,b_d,i]}{U[b_0,\ldots,b_d,i]}
                \EndIf
            \EndFor
        \EndFor
        \For{$d = h-1, \ldots, 0$}
            \Let{\Pi'[b_0,\ldots,b_d]}{\Pi'[b_0,\ldots,b_d] \prod_{i \in [c]} \left[\frac{L_{i}(X)-L_i(b_{d+1})}{X-b_{d+1}}\right]^{H(\mathcal{C}'[b_0,\ldots,b_d,i])-H(\mathcal{C}'[b_0,\ldots,b_d,i])}}
        \EndFor
        \Let{r'}{H(\mathcal{C}'[b_0,b_1],\ldots,\mathcal{C}'[b_0,\ldots,b_{h-1}], H(\mathcal{C}'[b_0,b_1]), \ldots, H(\mathcal{C}'[b_0,\ldots,b_h]),b_1,\ldots,b_h)}
        \Let{G'}{\prod_{j \in [h]} \Pi'[b_0,\ldots,b_j]^{r'^j}}
        \Let{t'}{H(r',G')}
        \Let{\pi'}{\prod_{j \in [h]} \Pi'[b_0,\ldots,b_j]^{\frac{r'^j}{t'-b_{j+1}}}}
        \State\Return $((\mathcal{C}'[b_0],\ldots,\mathcal{C}'[b_0,\ldots,b_{h-1}]), (G', \pi'))$
    \EndAlg
    \end{algorithmic}
\end{algorithm}

We now describe the update and proof update functions for Verkle trees (Algs.~\ref{alg.verkle.update} and~\ref{alg.verkle.proofupdate} respectively).
Since Verkle trees were proposed to support stateless clients,
we describe an update scheme that minimizes the runtime complexity of proof updates and does not require the users to download the updated messages or have access to old inner nodes.
As Verkle trees do not support sublinear update, we numerically estimate the size of the update information and the complexity of proof updates in Section~\ref{sec:evaluation}.

\subsection{Update Information}
Suppose the vector $(i, m_i)_{i \in [N]}$ is modified at some index $x$, $(b_1, \ldots, b_h) = \bin(x)$ to be $m'_x$.
Since each inner node is the hash of a KZG commitment, the new inner nodes $u'_{b_0,\ldots,b_j} = H(C'_{b_0,\ldots,b_j})$, $j \in [h]$, can be found as a function of the old commitments at the nodes and the powers of the Lagrange basis polynomials as described in Section~\ref{sec:kzg-update}:
\begin{IEEEeqnarray*}{C}
C'_{b_0,\ldots,b_h} \xleftarrow[]{} m'_x,\quad\quad C'_{b_0,\ldots,b_j} \xleftarrow[]{} C_{b_0,\ldots,b_j} [L_{b_{j+1}}]^{(u'_{b_0,\ldots,b_{j+1}}-u_{b_0,\ldots,b_{j+1}})}
\end{IEEEeqnarray*}
When $k$ messages are updated, the above calculation is repeated $k$ times for each update.

Update information $U$ consists of the new values of the KZG commitments on the path from the updated messages to the Verkle root akin to the Merkle trees, ordered in the canonical top-to-bottom and left-to-right order.

\subsection{Verkle Proofs}
Let $\pi_x$ denote the Verkle proof of some message $m_x$ at index $x$, $(b_1,\ldots,b_h) = \bin(x)$: $\pi_x = ((C_{b_0,b_1}, \ldots, C_{b_0,\ldots,b_{h-1}}), ([g(X)], \pi))$.
We define $\pi^f_x$ as the opening proof for index $x$ within polynomial $f$.
We observe that the commitment $[g(X)]$ and the proof $\pi$ can be expressed as functions of the opening proofs of the inner nodes $u_{b_0,b_1}, \ldots, u_{b_0,\ldots,b_h}$ at the indices $b_1,\ldots,b_h$ within the polynomials $f_{b_0}, \ldots, f_{b_0,\ldots,b_{h-1}}$, respectively.
Namely, $[g(X)]$ is
\begin{IEEEeqnarray*}{rCl}
\left[\sum_{j=0}^{h-1} {r}^j \frac{f_{b_0,\ldots,b_j}(X)-u_{b_0,\ldots,b_{j+1}}}{X-b_{j+1}}\right] &=& \prod_{j=0}^{h-1} \left[\frac{f_{b_0,\ldots,b_j}(X)-u_{b_0,\ldots,b_{j+1}}}{X-b_{j+1}}\right]^{r^j} \\
&=& \prod_{j=0}^{h-1} \left(\pi^{f_{b_0,\ldots,b_j}}_{b_{j+1}}\right)^{r^j}
\end{IEEEeqnarray*}
Similarly, the opening proof $\pi=\pi^{(h-g)}_t$ for index $t$ within the polynomial $h(X)-g(X)$ can be expressed as follows (see Appendix~\ref{sec:appendix-derivation} for details):
\begin{IEEEeqnarray*}{rCl}
\left[\frac{h(X)-g(X)-(h(t)-g(t))}{X-t}\right] &=& \prod_{j=0}^{h-1} \left[\frac{f_{b_0,\ldots,b_j}(X)-u_{b_0,\ldots,b_{j+1}}}{X-b_{j+1}}\right]^{\frac{r^j}{t-b_{j+1}}} \\
&=& \prod_{j=0}^{h-1} \left(\pi^{f_{b_0,\ldots,b_j}}_{b_{j+1}}\right)^{\frac{r^j}{t-b_{j+1}}}
\end{IEEEeqnarray*}

We assume that each user holding the Verkle proof $\pi_x$ for some index $x$, $(b_1,\ldots,b_h) = \bin(x)$, also holds the opening proofs $\pi^{f_{b_0,\ldots,b_j}}_{b_{j+1}}$, $j \in [h]$, \emph{in memory}.
As we will see in the next section, the user also holds the KZG commitments at the children of the inner nodes on the path from the root to the message $m_x$, \ie $C_{b_0,\ldots,b_j,i}$ for all $j \in [h]$ and $i \in [c]$ \emph{in memory}.
These opening proofs and KZG commitments are not broadcast as part of any proof; however, they are needed for the user to locally update its Verkle proof after message updates.

\subsection{Proof Update}

When the messages $(i_j, m_{i_j})_{j \in [k]}$ are updated to $(i_j, m'_{i_j})_{j \in [k]}$, to calculate the new Verkle proof $\pi'_x$, the user must obtain the new commitments $C'_{b_0}, \ldots, C'_{b_0,\ldots,b_{h-1}}$ on the path from the root to message $m_x$, the new commitment $[g'(X)]$ and the new opening proof $\pi'$ for the polynomial $h'(X)-g'(X)$ at index $t'= H(r',[g'(X)])$.
Message updates change the commitments at the inner nodes, which in turn results in new polynomials $f_{b_0,\ldots,b_j}$, $j \in [h]$.
Suppose each polynomial $f_{b_0,\ldots,b_j}$, $j \in [h]$, is updated so that
\begin{IEEEeqnarray*}{C}
f'_{b_0,\ldots,b_j}(X) = f_{b_0,\ldots,b_j}(X) + \sum_{i=0}^{c-1}(f'_{b_0,\ldots,b_j}(i)-f_{b_0,\ldots,b_j}(i)) L_{i}(X),
\end{IEEEeqnarray*}
where, by definition, $f'_{b_0,\ldots,b_j}(i)-f_{b_0,\ldots,b_j}(i) = u'_{b_0,\ldots,b_j,i}-u_{b_0,\ldots,b_j,i} = H(C'_{b_0,\ldots,b_j,i})-H(C_{b_0,\ldots,b_j,i})$.
Then, given the new and the old commitments $(C_{b_0,\ldots,b_j,i},C'_{b_0,\ldots,b_j,i})$ for $i \in [c]$ and $j \in [h]$, the table of Lagrange basis polynomials, and using the technique in Section~\ref{sec:kzg-update}, the new opening proofs $\tilde{\pi}^{f_{b_0,\ldots,b_j}}_{b_{j+1}}$ after the message updates can be computed as follows for $j \in [h]$:
\begin{IEEEeqnarray*}{C}
\tilde{\pi}^{f_{b_0,\ldots,b_j}}_{b_{j+1}} = \pi^{f_{b_0,\ldots,b_j}}_{b_{j+1}} \prod_{i=0}^{c-1}\left[\frac{L_{i}(X)-L_{i}(b_{j+1})}{X-b_{j+1}}\right]^{(H(C'_{b_0,\ldots,b_j,i})-H(C_{b_0,\ldots,b_j,i}))},
\end{IEEEeqnarray*}
where $\left[\frac{L_{i}(X)-L_{i}(b_{j+1})}{X-b_{j+1}}\right]$ is the opening proof of the Lagrange basis polynomial $L_i(X)$ at index $b_{j+1}$.
Once the new opening proofs are found, the new commitment $[g'(X)]$ and the new proof $\pi'$ become
\begin{IEEEeqnarray*}{C}
[g'(X)] = \prod_{j=0}^{h-1} \left(\tilde{\pi}^{f_{b_0,\ldots,b_j}}_{b_{j+1}}\right)^{{r'}^j}, \quad\quad \pi' = \prod_{j=0}^{h-1} \left(\tilde{\pi}^{f_{b_0,\ldots,b_j}}_{b_{j+1}}\right)^{\frac{{r'}^j}{t'-b_{j+1}}}
\end{IEEEeqnarray*}
where $r'=H(C'_{b_0,b_1},..,C'_{b_0,\ldots,b_{h-1}},u'_{b_0,b_1},..,u'_{b_0,\ldots,b_h},b_1,..,b_h)$ and
$t'=H(r',[g'(X)])$.
Note that both $r'$ and $t'$ can be calculated by the user given the new KZG commitments $C'_{b_0,\ldots,b_j,i}$ for all $i \in [c]$ and $j \in [h]$.

Finally, to retrieve the new KZG commitments $C'_{b_0,\ldots,b_j,i}$ for all $i \in [c]$ and $j \in [h]$, the user inspects the commitments published as part of the update information $U$:
$C'_{b_0,b_1,\ldots,b_{j-1},i} \xleftarrow[]{} U[b_0,b_1,\ldots,b_{j-1},i]$ if $U[b_0,b_1,\ldots,b_{j-1},i] \neq \bot$ and $C'_{b_0,b_1,\ldots,b_{j-1},i} \xleftarrow[]{} C_{b_0,b_1,\ldots,b_{j-1},i}$ otherwise, for all $i \in [c]$ and $j \in [h]$.

In Verkle trees, the user cannot calculate the effect of an updated message on an arbitrary inner node without the knowledge of the inner nodes on the path from the message to the target node.
For instance, suppose $U[b_0,b_1,\ldots,b_{j-1},i] = \bot$ for some $i \in [c]$ and $j \in [h]$, and the user wants to calculate the effect of an update from $m_x$ to $m'_x$ on $C'_{b_0,\ldots,b_{j-1},i,\tilde{b}_{j+1},\ldots,\tilde{b}_h}$, $\bin(x) = (b_1,\ldots,b_{j-1},i,\tilde{b}_{j+1},\ldots,\tilde{b}_h)$ and $\tilde{b}_j = i$.
Then, for each $\ell \in \{j,\ldots,h-1\}$, the user have to find
\begin{IEEEeqnarray*}{rCl}
C'_{b_0,\ldots,\tilde{b}_j,\ldots,\tilde{b}_h} &\xleftarrow[]{}& m'_x \\ C'_{b_0,\ldots,\tilde{b}_j,\ldots,\tilde{b}_\ell} &\xleftarrow[]{}& C_{b_0,\ldots,\tilde{b}_j,\ldots,\tilde{b}_\ell} [L_{\tilde{b}_{\ell+1}}]^{(u'_{b_0,\ldots,\tilde{b}_j,\ldots,\tilde{b}_{\ell+1}}-u_{b_0,\ldots,\tilde{b}_j,\ldots,\tilde{b}_{\ell+1}})},
\end{IEEEeqnarray*}
where $C'_{b_0,\ldots,\tilde{b}_j,\ldots,\tilde{b}_\ell}$ are the commitments on the path from the target commitment $C_{b_0,b_1,\ldots,b_{j-1},i}$ to the message $m_x$.
Hence, the user has to know the original commitments on the path from the message to the target commitment, \ie, keep track of inner nodes, which contradicts with the idea of stateless clients.
This shows the necessity of publishing all of the updated inner nodes as part of the update information.

\subsection{Complexity}
Suppose each KZG commitment is of size $|G|$ and each hash $H(C)$ of a KZG commitment, \ie each inner node, has size $|H|$.
Then, updating a single message results in one update at each level of the Verkle tree and requires $\Theta(h|H|)$ group operations.
Thus, when $k$ messages are updated, the new Verkle root can be found after $\Theta(kh|H|)$ group operations.
As $U$ consists of the published KZG commitments at the inner nodes and their indices, $|U| = \Theta(k \log_c{(N)}(\log{(N)}+|G|)) = \Tilde{\Theta}(k)|G|$, which implies $g_1(k) = k$.

The user can replace each KZG commitment at the children of the inner nodes from the root to its message in $\Theta(\log{(N)}+|G|)$ time by running a binary search algorithm over $U$.
Since there are $ch$ such commitments to be updated, \ie, $C_{b_0,\ldots,b_j,i}$, $i \in [c]$ and $j \in [h]$, updating these commitments takes $\Theta(c h (\log{(N)}+|G|)) = \Tilde{\Theta}(1)$ time.

Upon obtaining the new commitments $C'_{b_0,\ldots,b_{j-1},i}$, $i \in [c]$, $j \in [h]$, with access to the table of Lagrange basis polynomials, the user can update each opening proof $\pi_{b_{j+1}}$ (for the function $f_{b_0,\ldots,b_j}$),
$j \in [h]$, with $\Theta(c|H|)$ group operations.
Since there are $h$ such proofs, updating them all requires $\Theta(c h |H|)$ group operations.
Given the new proofs, computing the new commitment $[g'(X)]$ and proof $\pi'$ requires $\Theta(h |H|)$ group operations.
This makes the total complexity of updating a Verkle proof $\Theta(c h + 2 h) |H| T_G + \Theta(c h (\log_c{(N)}+|G|))$.
For a constant $c$ and $h = \log_c{(N)}$, this implies a worst-case time complexity of $\Tilde{\Theta}(1) |H| T_G$ for Verkle proof updates, \ie, $g_2(k) = 1$.

\subsection{A Concrete Evaluation}
\label{sec:evaluation}

We now estimate the size of the update information and the number of group operations to update a Verkle opening proof after observing an Ethereum block consisting of ERC20 token transfers.
As in Sections~\ref{sec:eval-homomorphic-tree} and~\ref{sec:eval-homomorphic}, suppose the block has the target size of $15$ million gas~\cite{gasandfees}, and each token transfer updates the balance of two distinct accounts stored at separate leaves of the Verkle tree.
Then, there are $\sim 230$ such transactions in the block, and the block updates $k = 460$ accounts.
We assume that the Verkle tree has degree $256$ (\cf~\cite{verkle}) and commits to $256^3$ accounts as in Section~\ref{sec:eval-homomorphic-tree}.
Then, each proof consists of $2$ KZG commitments, $C_{\bot,b_1}$ and $C_{\bot,b_1,b_2}$ and a multiproof consisting of the commitment $[g(X)]$ and proof $\pi'$.
These components are elements of the pairing-friendly elliptic curve $\mathrm{BLS}12\_381$ and consist of $|G| = 48$ bytes~\cite{verkle}.
This implies a proof size of $(\log_c{(N)}+1)|G| = 192$ bytes (excluding the message at the leaf and its hash value; adding those makes it $272$ bytes).

When $460$ accounts are updated, in the worst-case, the update information has to contain $k \log_c(N) (\log(N)+|G|) = 460 \times 3 \times (24+48)$ Bytes, \ie, $99.4$ kBytes. 
This is comparable to the size of the Ethereum blocks, which are typically below $125$ kBytes~\cite{block-size}.
Hence, even though the update information of Verkle trees is linear in $k$, it does not introduce a large overhead beyond the block data.
Note that the runtime of the proof updates are constant and do not scale in the number of updated messages $k$, or the Ethereum block size.

On the other hand, in the worst case, an opening proof can be updated after $c \log{(c)} + 2 \log_c{(N)}$ group exponentiations.
Following the benchmarks published in \cite{blst-benchmark} for the specified curve, these operations take $(c + 2) \log_c{(N)}\ 0.000665471\text{ s} = 0.52$ seconds on commodity hardware, given a runtime of $665471$ nanoseconds per exponentiation of a group element with a message hash value.
This is again comparable to the $12$ second inter-arrival time of Ethereum blocks.

Table~\ref{tab:comparison} compares the Verkle proof size $|\pi| = (\log_c{(N)}+1) |G|$, update information size $|U| = k \log_c(N) (\log_c{N}+|G|)$, the upper bound $(c + 2) \log_c{N}$ on the number of group exponentiations needed for a single proof update and the estimated time it takes to do these operations on a commodity hardware for different values of $c$, the Verkle tree degree, while keeping the number of accounts and the updated accounts fixed at $2^{24}$ and $460$ respectively.
The table shows the trade-off between the Verkle proof and update information size on one size and update complexity on the other.
\begin{table}[]
\centering
\begin{tabular}{|c|c|c|c|c|}
 \hline
  $c$ & $|\pi|$ (Bytes) & $|U|$ (kBytes) & $\#$ Group Exp. & Time (s) \\  [0.5ex]
 \hhline{|=|=|=|=|=|}
  $2$ & 1200 & 794.9 & 96 & 0.064 \\
  \hline
  $4$ & 628 & 397.4 & 72 & 0.048 \\
  \hline
  $16$ & 336 & 198.7 & 108 & 0.072 \\
  \hline
  $64$ & 240 & 132.5 & 264 & 0.18 \\
  \hline
  $256$ & 192 & 99.4 & 774 & 0.52 \\
 \hline
\end{tabular}
\vspace{0.5cm}
\caption{For different values of the tree degree~$c$, the table shows the
Verkle proof size which is $|\pi| = (\log_c{(N)}+1)|G|$;
the update information size which is $|U| = k \log_c{(N)} (\log{(N)}+|G|)$;
the number of group exponentiations for a single proof update which is $(c + 2) \log_c{(N)}$; 
and the estimated time for a single proof update. 
We use $N = 2^{24}$ accounts in total, $k=460$ updates at the accounts, and
a group element size of $|G|=48$ bytes.
}
\vspace{-0.5cm}
\label{tab:comparison}
\end{table}

Finally, we calculate the total size of the public parameters required for updating Verkle trees.
For this purpose, we find the total size of the KZG commitments of Lagrange basis polynomials and the opening proofs for each Lagrange basis polynomial at each index $i \in [c]$.
Hence, the total size of the public parameters becomes $(c + c^2)|G|$, which is equal to $3.16$ MBytes.

Comparing Tables~\ref{tab:comparison} and~\ref{tab:comparison-homomorphic-tree} shows that the update information size and the proof update time of Verkle trees at $c = 256$, the parameter used by the Ethereum implementation~\cite{go-ethereum-verkle}, are outperformed by the homomorphic tree construction using $\nu = 1/2$.
Despite this, Verkle trees with a large degree $c = 256$ have small opening proofs of size $192$ Bytes, whereas the proof size of the AMT construction stands at $1.2$ kBytes.
In terms of the public parameter size, Verkle trees also outperform homomorphic tree constructions by around a factor of $10000$.
Whereas the proof size of AMTs scale as $O(\log{(N)})$ in the number $N$ of messages, proof size of Verkle trees remain constant at size $O(1/\epsilon)$, when their degree grows as $\Theta(N^\epsilon)$ for some fixed $\epsilon > 0$.
However, as $N$ grows, larger degree implies larger public parameters for Verkle trees.

Comparing Table~\ref{tab:comparison} with Table~\ref{tab:comparison-homomorphic} shows that a Verkle tree with any given degree $c$, $1 < c \leq 256$, significantly outperforms the existing homomorphic Merkle trees in terms of almost all of proof size, update information size and proof update time.
However, unlike Verkle trees, homomorphic Merkle trees are post-quantum secure, and do not require trusted setup or public parameters.

\section{Lower Bound}
\label{sec:lower-bound}

Finally, we prove the optimality of our VC scheme with sublinear update by proving a lower bound on the size of the update information given an upper bound on the complexity of proof updates.
The lower bound is shown for VCs that satisfy the following \emph{\proofbinding} property.
It formalizes the observation that for many dynamic VCs (\eg, Merkle trees \cite{merkle}, Verkle trees \cite{verkle}, KZG commitments \cite{kate}, RSA based VCs \cite{rsa-vc}) including homomorphic Merkle trees, the opening proof for a message at some index can often act as a commitment to the vector of the remaining messages.
\begin{definition}
\label{def:proof-binding}
A VC scheme is said to be \emph{\proofbinding} if the following probability is negligible in $\lambda$ for all PPT adversaries $\mathcal{A}$:
\begin{IEEEeqnarray*}{C}
\Pr\left[\subalign{\textsc{Verify}_{pp}(C, m_{i^*}, i^*, \pi) &= 1 \\ \land \textsc{Verify}_{pp}(C', m_{i^*}, i^*, \pi) &= 1} \mathbf{\colon} \subalign{pp \xleftarrow{} \textsc{KeyGen}(1^\lambda, N);& \\ \pi, m_{i^*}, (m_0, \ldots, m_{i^*-1}, m_{i^*+1}, \ldots, m_{N-1}),& \\ (m'_0, \ldots, m'_{i^*-1}, m'_{i^*+1}, \ldots, m'_{N-1}) \xleftarrow{} \mathcal{A}(pp);& \\ (m_0, \ldots, m_{i^*-1}, m_{i^*+1}, \ldots, m_{N-1})& \\ \neq (m'_0, \ldots, m'_{i^*-1}, m'_{i^*+1}, \ldots, m'_{N-1});& \\ \textsc{Commit}_{pp}(m_0, \ldots, m_{i^*-1}, m_{i^*}, m_{i^*+1}, \ldots, m_{N-1}) = C;& \\ \textsc{Commit}_{pp}(m'_0, \ldots, m'_{i^*-1}, m_{i^*}, m'_{i^*+1}, \ldots, m'_{N-1}) = C'&}\right]
\end{IEEEeqnarray*}
\end{definition}
In the definition above, the opening proof while committing to the vector of remaining messages, need not be of the same format as the original VC.
For instance, for RSA accummulators, the opening proof is itself an RSA accummulator, whereas for Merkle trees, the opening proof is not a Merkle tree root, but contains a sequence of inner nodes and the index of the opened message.
Nevertheless, the proof and the message together act as a commitment to the vector of remaining messages.

The \proofbinding property is distinct from the position-binding (security) property of VCs: whereaas position-binding states the difficulty of opening the VC to two different messages at the same index, \proofbinding implies the difficulty of creating two VCs with different messages that open to the \emph{same} message at some index $i$ with the exact \emph{same} proof.

We next state the main lower bound for \proofbinding VCs.
\begin{theorem}
\label{thm:lower-bound}
Consider a dynamic and \proofbinding VC such that for every PPT adversary $\mathcal{A}$, it holds that 
\begin{IEEEeqnarray*}{C}
\Pr\left[\substack{\textsc{Verify}_{pp}(C, m, i, \pi_i) = 1 \land \\ \textsc{Verify}_{pp}(C, m', i, \pi'_i) = 1\ \land\ m \neq m'}\ \colon \substack{pp \xleftarrow{} \textsc{KeyGen}(1^\lambda, N) \\ (C, m, m', \pi_i, \pi'_i) \xleftarrow{} \mathcal{A}(pp)}\right] \leq e^{-\Omega(\lambda)}.
\end{IEEEeqnarray*}
Then, for this VC, if $g_2(k)  = O(k^{1-\nu})$, then $g_1 = \Omega(k^\nu)$ for all $\nu \in (0,1)$.
\end{theorem}
To prove the lower bound, we first show that \proofbinding implies that $(i^*, m_{i^*}, \pi)$ is a binding commitment to the rest of the vector.
\begin{lemma}
\label{lem:proof-binding}
Consider a dynamic and \proofbinding VC, where $\pi$ is the correctly generated opening proof for the message $m_i$ at some index $i$.
Then, for any $i \in [N]$, it holds that the tuple $(i, m_i, \pi)$ is a binding commitment to the vector of messages $m_j$, $j \in [N]$, $j \neq i$.
\end{lemma}
\begin{proof}
Since the VC is \proofbinding, with overwhelming probability, no PPT adversary $\mathcal{A}$ can find an opening proof $\pi^*$, an index $i^*$, a message $m^*$ and two sequences of messages 
such that 
\begin{IEEEeqnarray*}{C}
(m_1, \ldots, m_{i^*-1}, m_{i^*+1}, \ldots, m_{N-1}) \neq(m'_1, \ldots, m'_{i^*-1}, m'_{i^*+1}, \ldots, m'_{N-1})
\end{IEEEeqnarray*}
and $\textsc{Verify}_{pp}(C, m_{i^*}, i^*, \pi) = \textsc{Verify}_{pp}(C', m_{i^*}, i^*, \pi) = 1,$
where $C$ and $C'$ are commitments to the message sequences $(m_1, \ldots, m_{i^*-1}, m_{i^*}, m_{i^*+1}, \ldots, m_{N-1})$ and $(m'_1, \ldots, m'_{i^*-1}, m_{i^*}, m'_{i^*+1}, \ldots, m'_{N-1})$.
Thus, it holds that the tuple $(i, m_i, \pi)$ is a binding commitment to the vector of messages $m_j$, $j \in [N]$, $j \neq i$, with the following new commitment function:
\begin{IEEEeqnarray*}{C}
\textsc{NewCommit}_{pp}((m_j)_{j \in [N], j \neq i}) = (i, m_i, \textsc{Open}_{pp}(m_i, i, \aux)),
\end{IEEEeqnarray*}
where $\aux = \textsc{Commit}_{pp}(m_0, \ldots, m_{N-1}).\aux$.
\end{proof}
The following lemma shows that all randomized VCs can be derandomized to obtain a deterministic and secure VC as we do not use hiding commitments in this work.
\begin{lemma}
\label{lem:derandomization}
Consider a VC $\PI$, where the commitment is a random function of the public parameters $pp$ and the committed messages.
Let $\PI'$ denote the VC that is the same as $\PI$, except that the randomness is fixed.
Then, $\PI'$ is a correct and secure VC with at most the same upper bound on the error probability.
\end{lemma}
\begin{proof}
Let $R$ denote the sequence of bits sampled uniformly at random from the set $\mathcal{R}$ to instantiate the VC $\PI$.
Since $\PI$ is binding, no PPT adversary $\mathcal{A}$ can find two different sequences of messages $\mathbf{m}$ and $\mathbf{m'}$ such that $\PI(\mathbf{m}, R) = \PI(\mathbf{m'}, R')$ for some $R,R' \in \mathcal{R}$, except with negligible probability.
This implies that for any fixed $R^* \in \mathcal{R}$, no PPT adversary $\mathcal{A}$ can find two different sequences of messages $\mathbf{m}$ and $\mathbf{m'}$ such that $\PI(\mathbf{m}, R^*) = \PI(\mathbf{m'}, R^*)$, except with negligible probability.
Hence, the commitment scheme $\PI'(.) = \PI(., R^*)$ is a position-binding, \ie, secure VC. 
Its correctness follows from the correctness of $\PI$.
\end{proof}
Finally, equipped with these lemmas, we can prove Theorem~\ref{thm:lower-bound} for dynamic and \proofbinding VCs.
\begin{proof}[Proof of Theorem~\ref{thm:lower-bound}]
Suppose the messages $m_{i_j}$, $j \in [k]$, are updated to $m'_{i_j}$.
Define $\mathcal{S}$ as the sequence $(m'_{i_j})_{j \in [k]}$, and let $m'_i = m_i$ for $i \notin \{i_j \colon j \in [k]\}$.
Let $\mathcal{P}_i$, $i \in [N]$, denote the user that holds the opening proof $\pi_i$ for the message $m_i$ at index $i$, and aims to calculate the new proof $\pi'_i$ for the message $m'_i$ using $\pi_i$, the update information $U$ and the old and the new sequences of messages $m_i, m'_i$, $i \in [N]$.
Suppose $g_2 = O(k^{1-\nu})$.
Then, there exists a constant $\alpha$ such that each user can read at most $\alpha k^{1-\nu}$ of the updated messages while updating its opening proof.
Let $\mathcal{S}_i \subseteq (m'_{i_j})_{j \in [k]}$ denote the sequence of updated messages and their indices, which were \emph{not} observed by $\mathcal{P}_i$, and $\overline{\mathcal{S}}_i = \mathcal{S} \setminus \mathcal{S}_i$ denote the sequence read by $\mathcal{P}_i$.
Here, $|\mathcal{S}|$ denotes the number of messages within the sequence $\mathcal{S}$.
Since $\mathcal{P}_i$ is assumed to know $m'_i$, it must be that $m'_i \in \overline{\mathcal{S}}_i$. 

We next show that each user $\mathcal{P}_i$ that successfully updates its opening proof must download enough bits of $U$ to generate a binding, deterministic commitment to the set $\mathcal{S}_i$.
By Lemma~\ref{lem:proof-binding}, the tuple $(i, m'_i, \pi'_i)$ is a binding commitment to the sequence of messages $(m'_j)_{j \in [N], j \neq i}$.
This implies that the tuple $(i, \overline{\mathcal{S}}_i, \pi'_i)$ is a binding commitment to the sequence $\mathcal{S}_i$. 
By Lemma~\ref{lem:derandomization}, the commitment $(i, \overline{\mathcal{S}}_i, \pi'_i)$ can be de-randomized to obtain a deterministic commitment $C_i$ to the sequence $\mathcal{S}_i$ (with at most the same upper bound on the error probability).
Let $\PI$ denote the deterministic VC scheme such that $C_i = \PI(\mathcal{S}_i)$. 
Since $\PI$ is a deterministic function given the public parameters, and the updated messages are sampled independently and uniformly at random, then $I(\mathcal{S}_i;\{m_i\}_{i \in N},\overline{\mathcal{S}}_i|pp) = 0$, where $I(.;.)$ is the mutual information.
Moreover, as $\pi_i$ is a function of the old messages $\{m_i\}_{i \in N}$ and the randomness of the original VC, $I(C_i; \pi_i|pp) = 0$.
Hence, $C_i = f(U, i, \{m_i\}_{i \in N}, \pi)$ is a deterministic function of the update information $U$.

For all $i \in [k]$, it holds that $|\mathcal{S}_i| \geq k - \alpha k^{1-\nu}$ and $m'_i \notin \mathcal{S}_i$.
Given these constraints, the minimum number of distinct sequences $\mathcal{S}_i$ is $\frac{k}{\alpha k^{1-\nu}} = \frac{k^\nu}{\alpha}$.
For an appropriately selected $\beta$ that will be defined later, without loss of generality, let $\mathcal{S}_0, \ldots, \mathcal{S}_{M-1}$ denote the first 
\begin{IEEEeqnarray*}{C}
M = \min\left(\left\lfloor \frac{k^\nu}{\beta} - \frac{\alpha}{\beta} - \frac{\lambda}{\beta k^{1-\nu}} \right\rfloor, \frac{k^\nu}{\alpha}\right)
\end{IEEEeqnarray*}
distinct sequences.
Since $C_i$ is a deterministic function of $U$ for all $i \in N$, it holds that the Shannon entropy $H(.)$ of $U$ satisfies the following expression:
\begin{IEEEeqnarray*}{C}
H(U) \geq H(C_0, \ldots, C_{M-1})
\geq H(C_0) + \sum_{i=1}^{M-1} H(C_i | C_0, \ldots, C_{i-1})
\end{IEEEeqnarray*}
As $g_2(k) = O(k^{1-\nu})$, there exists a constant $\beta$ such that each user can download at most $\beta k^{1-\nu}$ bits of data from $U$.
Then, for all $i \in [k]$, it must be that $H(C_i) \leq H(U) \leq \beta k^{1-\nu}$ since $C_i$ is a deterministic function of $U$ for each $i \in [N]$.

Finally, we show that $H(C_0), H(C_i | C_0, \ldots, C_{i-1}) = \Omega(\lambda)$ for all $i=1, \ldots, M-1$.
Towards contradiction, suppose $\exists i^* \colon H(C_{i^*} | C_0, \ldots, C_{i^*-1}) = o(\lambda)$.
Note that
\newline
\begin{IEEEeqnarray*}{rCl}
H(C_0, \ldots, C_{i^*-1}) &\leq& \sum_{i=0}^{M-1} H(C_i) \\
&\leq& \min\left(\frac{k^\nu}{\beta} - \frac{\alpha}{\beta} - \frac{\lambda}{\beta k^{1-\nu}}, \frac{k^\nu}{\alpha}\right) \beta k^{1-\nu} \leq k-\alpha k^{1-\nu}-\lambda.
\end{IEEEeqnarray*}
Now, consider an adversary $\mathcal{A}$ that tries to break the binding property of the VC scheme $\PI$.
Due to the upper bound on the entropy of $(C_0, \ldots, C_{i^*-1})$, it holds that
$H(\mathcal{S}_{i^*} | C_0, \ldots, C_{i^*-1}) \geq \lambda$; 
since $H(\mathcal{S}_{i^*}) \geq k-\alpha k^{1-\nu}$, and
\begin{IEEEeqnarray*}{rCl}
&& H(\mathcal{S}_{i^*}) - H(\mathcal{S}_{i^*} | C_0, \ldots, C_{i^*-1}) \\
&=& I(\mathcal{S}_{i^*}; (C_0, \ldots, C_{i^*-1})) \leq H(C_0, \ldots, C_{i^*-1}) \leq k-\alpha k^{1-\nu}-\lambda.
\end{IEEEeqnarray*}

However, when $H(C_{i^*} | C_0, \ldots, C_{i^*-1}) = o(\lambda)$, for sufficiently large $\lambda$, given $(C_0, \ldots, C_{i^*-1})$, the adversary can find a collision such that $\PI(\mathcal{S}_{i^*})=\PI(\mathcal{S}'_{i^*})$ for two $\mathcal{S}_{i^*} \neq \mathcal{S}'_{i^*}$, with probability $2^{-o(\lambda)}$.
As this is a contradiction, it must be that $H(C_0)$ and $H(C_i | C_0, \ldots, C_{i-1}) = \Omega(\lambda)$ for all $i < M$, and thus, $H(U) = \Omega(k^\nu \lambda)$ and $g_1(k) = \Omega(k^\nu)$.
\end{proof}
\begin{remark}
\label{rem:1}
Theorem~\ref{thm:lower-bound} shows that the update information length scales as $\Tilde{\Theta}(k^\nu \lambda)$ when the runtime complexity for proof updates is $\Tilde{\Theta}(k^{1-\nu})$ and the error probability for the security of the VC is $e^{-\Omega(\lambda)}$ for a PPT adversary.
When the error probability is just stated to be negligible in $\lambda$, then the same proof can be used to show that the update information length must scale as $\Omega(k^\nu \polylog(\lambda))$ for any polynomial function of $\log(\lambda)$.
\end{remark}
To demonstrate the optimality of homomorphic Merkle trees, we show that, like many other VCs, they satisfy the \proofbinding property:
\begin{theorem}
\label{lem:proof-binding-hmt}
The \proofbinding property is satisfied by secure homomorphic Merkle trees.
\end{theorem}
\begin{proof}
Consider an adversary $\mathcal{A}$ that finds an opening proof $\pi$, an index $i^*$, a message $m_{i^*}$, and sequences of messages $(m_0, \ldots, m_{i^*-1}, m_{i^*+1}, \ldots, m_{N-1}) \neq (m'_0, \ldots, m'_{i^*-1}, m'_{i^*+1}, \ldots, m'_{N-1})$ such that $\textsc{Verify}_{pp}(C, m_{i^*}, i^*, \pi) = 1$ as well as $\textsc{Verify}_{pp}(C', m_{i^*}, i^*, \pi) = 1$ for the commmitments $C$ and $C'$ to the sequences $(m_0, \ldots, m_{i^*-1}, m_{i^*}, m_{i^*+1}, \ldots, m_{N-1})$ and $(m'_0, \ldots, m'_{i^*-1}, m_{i^*}, m'_{i^*+1}, \ldots, m'_{N-1})$.
\newline
We show how it can break the position-binding property with this knowledge.
The adversary first constructs the two homomorphic Merkle trees committing to these sequences.
It then finds the first inner node $u'$ within the proof $\pi$ on the path from the message $m_{i^*}$ to the root such that the subtrees under this inner node contain different sequences of messages $(m_{a}, \ldots, m_{b}) \neq (m'_{a}, \ldots, m'_{b})$ at the leaves.
By definition $g^{-1}(b')$ is the homomorphic Merkle tree commitment to the sequence of messages under the node $b'$ on both trees.
Thus, the adversary has created two homomorphic Merkle trees committing to different sequences of messages but with the same homomorphic Merkle tree commitment $C = g^{-1}(b')$.
This implies that the adversary can find a tuple $(C, m_i \in (m_a, \ldots, m_b), m'_i \in (m'_a, \ldots, m'_b), \pi_i, \pi'_i)$ such that $\textsc{Verify}_{pp}(C, m_i, i, \pi_i) = 1$ and $\textsc{Verify}_{pp}(C, m'_i, i, \pi'_i) = 1$ and $m \neq m'$.
Since for all adversaries $\mathcal{A}$, the probability that $\mathcal{A}$ finds such a tuple is negligible in the security parameter $\lambda$, for all adversaries $\mathcal{A}$, the probability that $\mathcal{A}$ finds an opening proof $\pi$, index $i^*$, message $m_{i^*}$ and sequences of messages $(m_0, \ldots, m_{i^*-1}, m_{i^*+1}, \ldots, m_{N-1}) \neq (m'_0, \ldots, m'_{i^*-1}, m'_{i^*+1}, \ldots, m'_{N-1})$ with the above qualities is negligible in $\lambda$, implying that homomorphic Merkle trees are \proofbinding.
\end{proof}

\section{Conclusion}
\label{sec:conclusion}

Dynamic VCs with sublinear update are the key to reducing the size of the global update information while minimizing the runtime of clients synchronizing with the latest commitment.
In this work, we propose a construction that can achieve an update information size of $\Theta(k^\nu)$ and a proof update time of $\Theta(k^{1-\nu})$ in the number of changed messages $k$.
Our construction combines a novel update algorithm (Alg.~\ref{alg.vc.sublinear.p}) with homomorphic trees  and homomorphic Merkle trees that allow each inner node to be expressed as a homomorphic or linear function of the underlying messages.
It achieves the smallest asymptotic complexity for the update information size and proof update time.
We also provide update algorithms for the Verkle trees proposed for stateless clients on Ethereum.

The existing instantiations of homomorphic Merkle trees are based on lattices and require relatively large parameters for security.
In contrast, the existing instantiations of homomorphic trees are based on pairings, and despite outperforming Verkle trees, require trusted setup and relatively large public parameters and opening proofs.
As such, designing post-quantum secure, asymptotically optimal and practically efficient dynamic VCs without trusted setup and with small opening proofs remains an open problem.

\bigskip
{\bf Acknowledgments.}
This work was partially funded by NSF, DARPA, the Simons Foundation, and NTT Research.
Additional support was provided by the Stanford Center for Blockchain Research.
Opinions, findings, and conclusions or recommendations
expressed in this material are those of the authors and do not
necessarily reflect the views of DARPA.

\bibliographystyle{plainurl}
\bibliography{references}

\appendix

\section{Lower Bound on the Size of the Update Information}
\label{sec:appendix}

\begin{theorem}
\label{thm:lower-bound-2}
Consider a dynamic accumulator, where $k$ out of $N$ messages $m_{i_j}$ are updated to $m'_{i_j} \neq m_{i_j}$, $j \in [k]$.
Suppose $|M| = \poly(\lambda)$.
Then, $\Omega(k \log{(N|\mathcal{M}|}))$ bits of information must be published to enable updating the opening proofs after these $k$ updates.
\end{theorem}
\begin{proof}
The proof idea is very similar to those presented in~\cite{impossibility,revocable}. 
Namely, the update information must contain a minimum amount of bits for the VC to remain correct and secure after the update.

Consider a game between a platform $\mathcal{P}$ maintaining the data structures of the VC and an adversary $\mathcal{A}$.
The platform $\mathcal{P}$ updates $k$ out of $N$ messages $m_{i,j}$ to $m'_{i,j} \neq m_{i_j}$, $j \in [k]$, in a way not known to user $\mathcal{A}$, and publishes the update information $U$ along with the new commitment value $C'$ (let $m'_i = m_i$ for $i \notin \{i_j \colon j \in [k]\}$).
Before receiving the update information, $\mathcal{A}$ knows the old sequence of messages $m_i$, $i \in [N]$, and their opening proofs $\pi_i$.
Upon receiving the update information, $\mathcal{A}$ updates the opening proofs for each message to $\pi'_i$.
Then, it must be that for all $j \in [k]$, $\textsc{Verify}_{pp}(C', m'_{i_j}, i_j, \pi'_{i_j}) = 1$, and for all $i \notin \{i_j \colon j \in [k]\}$, $\textsc{Verify}_{pp}(C', m_i, i, \pi'_i) = 1$. 
Otherwise there would be messages among $m'_i$, $i \in [N]$, for which an updated witness cannot be computed, violating correctness.
Similarly, for all $j \in [k]$, $\textsc{Verify}_{pp}(C', \tilde{m}, i_j, \pi'_{i_j}) = 0$ for any $\tilde{m} \neq m'_{i_j}$; as otherwise the position binding property, thus security, would be violated.
Hence, by calling the function $\textsc{Verify}_{pp}(C', m_i, i, \pi'_{i_j})$ for each index, $\mathcal{A}$ can figure out the indices $i_j$, $j \in [k]$, where the messages were updated.
Similarly, by evaluating the function $\textsc{Verify}_{pp}(C', \tilde{m}, i_j, \pi'_{i_j})$ for the $|\mathcal{M}|$ possible messages $\tilde{m}$ for each $j \in [k]$, $\mathcal{A}$ can identify the new value $m'_{i_j}$ of the message at each such index $i_j$.
Hence, the adversary can recover the sequence $(i_j, m'_{i_j})_{j \in [k]}$.

As there are $\frac{N!}{(N-k)!}\log^k{|\mathcal{M}|}$ possible sequences $(i_j, m'_{i_j})_{j \in [k]}$, it holds that 
$$|U| \geq \log{\left(\frac{N!}{(N-k)!}\log^k{|\mathcal{M}|}\right)} = \Omega(k (\log{N}+\log{|\mathcal{M}|})).$$
\end{proof}

\begin{remark}
When $|M| = \Omega(\poly(\lambda))$, the minimum number of bits to be published depends on the error probability for the security of the PPT adversary. 
As in Remark~\ref{rem:1}, if $|\mathcal{M}| = \Theta(2^{\lambda})$ and the error probability is $e^{-\Omega(\lambda)}$, then the same proof can be used to show that $\Tilde{\Omega}(k\lambda)$ bits of information must be published.
When the error probability is just stated to be negligible in $\lambda$, then the number of bits must scale as $\Tilde{\Omega}(k\polylog{\lambda})$ for any polynomial function of $\log{(\lambda)}$.
\end{remark}

\section{Homomorphic Trees with No Update Information}
\label{sec:homomorphic-no-update-information}

\begin{algorithm}[ht!]
    \captionsetup{font=small} 
    \caption{Algorithms for a homomorphic tree. Each user knows the total number of leaves $N$. 
    }
    \label{alg.merkle.homomorphic}
    \begin{algorithmic}[1]\footnotesize
    \Alg{\sc Update}{$C, (i, m'_i-m_i)_{i \in [N]}, \T$}
        \Let{C'}{\mathrm{new}\ \mathrm{VC}} \Comment{Commitment need not be the same as the root.}
        \State\Return $(C',\emptyset, \T)$
    \EndAlg
    \Alg{\sc ProofUpdate}{$(i_j, m'_{i_j}-m_{i_j})_{j \in [k]}, \pi_x , m'_x, x, U$}
        \Let{\pi'_x}{\{\}}
        \Let{(b_1,\ldots,b_h)}{\bin(x)}
        \For{$d = h, \ldots, 1$}
            \Let{\mathcal{S}}{\{j \in [k] \colon\bin(i_j)[0:d-p-1] = (b_1, \ldots, b_{d-p})\}}
            \Let{\pi'_x[b_0,\ldots,b_d]}{\pi_x[b_0,\ldots,b_d] \cdot \prod_{j \in \mathcal{S}} h_{i_j,d, (b_{d-p+1},b_d)}(m'_{i_j}-m_{i_j})}
        \EndFor
        \State\Return $\pi'_x$
    \EndAlg
    \end{algorithmic}
\end{algorithm}

\subsection{Update Information}
When $k$ messages are updated, the update algorithm merely calculates the new value of the vector commitment.
As in KZG commitments, the update information is $U=\emptyset$.

\subsection{Proof Update}
When the messages are modified at $k$ points, each user holding an opening proof $\pi_x$ for index $x$ can calculate the new values of the inner nodes within the proof using the old and the new messages and modify the proof respectively.

\subsection{Complexity}
Calculating each partial digest $h_{x,j,c}$ takes constant time.
Then, each user can update each inner node within its opening proof after at most $k$ operations, making the total number of operations $\Theta(k\log{(N)}) = \Tilde{\Theta}(k)$ in the worst case.
The size of the update information $U$ is $\Tilde{\Theta}(1)$.
Hence, this scheme matches the algebraic VCs in terms of complexity.

\section{Why Are Homomorphic Merkle Trees Needed?}
\label{sec:why-homomorphic}

Merkle trees based on SHA256 can also achieve complexity sublinear in $k$, for both the update information and the runtime of proof updates, if the users have access to the old messages and inner nodes of the Merkle tree.
In this case, homomorphism is not needed since the nodes can find the effect of the updated messages on the inner nodes within their Merkle proofs by hashing these messages together with the old inner nodes.
However, this is possible for only a \emph{single} batch of updates.
Indeed, if this scheme is to be repeated, the assumption of having access to the old inner nodes requires the users to keep track of changes throughout the Merkle tree, by calculating the effect of all updated messages on all inner nodes.
This implies a runtime linear in $k$ per proof updates.
In contrast, homomorphic Merkle trees can maintain a sublinear complexity for future proof updates since they do not require access to the old messages and inner nodes for finding the partial digests of the updated messages.

\section{An Alternative Construction} 
\label{sec:appendix-alternative}

An alternative tree-based VC is proposed by \cite{PPS21}, where each inner node is itself a lattice-based VC to its children (akin to Verkle trees~\cite{verkle}).
Opening proof for a message consists of the inner nodes (commitments) on the path from the message to the root, along with the opening proofs for these inner nodes with respect to their parent nodes.
The construction again enables expressing each inner node as a sum of partial digests of the messages underneath.
Using the public parameters and the updated inner nodes, users can then derive their updated opening proofs at different heights of the tree.
This construction supports trees of large degrees $c$ without a linear increase in the proof size as would be the case for Merkle trees; this however comes at the cost of a larger runtime complexity for proof updates, proportional to the degree.
Section~\ref{sec:verkle-update} describes similar steps in the context of Verkle trees, and exposes the dependence of the runtime complexity of proof updates on the tree degree $c$.

\section{Derivation of the Opening Proof $\mathbf{\pi^{(h-g)}_t}$} 
\label{sec:appendix-derivation}

Since
\begin{IEEEeqnarray*}{rCl}
h(X)-g(X) &=& \sum_{j=0}^{h-1} r^j \left(\frac{f_{b_0,\ldots,b_j}(X)}{t-b_{j+1}} - \frac{f_{b_0,\ldots,b_j}(X)-u_{b_0,\ldots,b_{j+1}}}{X-b_{j+1}}\right) \\
&=& \sum_{j=0}^{h-1} r^j \frac{(X-t)f_{b_0,\ldots,b_j}(X)+u_{b_0,\ldots,b_{j+1}}(t-b_{j+1})}{(t-b_{j+1})(X-b_{j+1})},
\end{IEEEeqnarray*}
the opening proof $\pi=\pi^{(h-g)}_t$ for index $t$ within the polynomial $h(X)-g(X)$ is
\begin{IEEEeqnarray*}{rCl}
&& \left[\frac{h(X)-g(X)-(h(t)-g(t))}{X-t}\right] \\
&=& \left[\sum_{j=0}^{h-1} \frac{r^j}{X-t} \left(\frac{(X-t)f_{b_0,\ldots,b_j}(X)+u_{b_0,\ldots,b_{j+1}}(t-b_{j+1})}{(t-b_{j+1})(X-b_{j+1})}-\frac{u_{b_0,\ldots,b_{j+1}}}{t-b_{j+1}}\right) \right] \\
&=& \left[\sum_{j=0}^{h-1} \frac{r^j}{t-b_{j+1}} \frac{f_{b_0,\ldots,b_j}(X)-u_{b_0,\ldots,b_{j+1}}}{X-b_{j+1}}\right] \\
&=& \prod_{j=0}^{h-1} \left[\frac{f_{b_0,\ldots,b_j}(X)-u_{b_0,\ldots,b_{j+1}}}{X-b_{j+1}}\right]^{\frac{r^j}{t-b_{j+1}}} = \prod_{j=0}^{h-1} \left(\pi^{f_{b_0,\ldots,b_j}}_{b_{j+1}}\right)^{\frac{r^j}{t-b_{j+1}}}
\end{IEEEeqnarray*}

\section{Update and Proof Update Algorithms for KZG Commitments and Merkle Trees}
\label{sec:appendix-algorithms}

\begin{algorithm}[ht!]
    \captionsetup{font=small} 
    \caption{Update and proof update algorithms for KZG commitments. Each user knows the value $N$ and has access to a table of commitments and opening proofs of Lagrange basis polynomials $L_i(X)$ at each point $j \in \mathbb{F}_p$, $j<N$. The parameter $\pi_{i_j,x}$ denotes the opening proof of $L_i(X)$ at index $i_j$. The update and proof update algorithms only need the difference of the updated messages.}
    \label{alg.kzg}
    \begin{algorithmic}[1]\footnotesize
    \Alg{\sc Update}{$C, (i_j, m'_{i_j}-m_{i_j})_{j \in [k]}, \emptyset$}
        \Let{C'}{C \prod_{j=1}^m[L_{i_j}(X)]^{(m'_{i_j}-m_{i_j})}}
        \Let{U}{\emptyset}
        \State\Return $(C',U, \emptyset)$
    \EndAlg
    \Alg{\sc ProofUpdate}{$C, (i_j, m'_{i_j}-m_{i_j})_{j \in [k]}, \pi_{x} , m'_{x}, x, U$}
        \Let{\pi'_x}{\pi_x \prod_{j=1}^k \pi^{m'_{i_j}-m_{i_j}}_{i_j,x}}
        \State\Return $\pi'_x$
    \EndAlg
    \end{algorithmic}
\end{algorithm}

\begin{algorithm}[ht!]
    \captionsetup{font=small} 
    \caption{Update and proof update algorithms for a Merkle tree. Each user knows the total number of leaves $N$. The variable $\T$ denotes the dictionary of the old Merkle tree nodes, prior to the message updates, and $\pi_x$ denotes the Merkle proof for index $x$, whose components are organized as a dictionary. The algorithm $\textsc{Empty}(.)$ initializes a dictionary with the indices of the Merkle tree nodes as keys and $\bot$ as their values. Not all of $U$ is passed to the proof update algorithm and its relevant parts are read selectively to keep the runtime at a minimum.}
    \label{alg.merkle}
    \begin{algorithmic}[1]\footnotesize
    \Alg{\sc Update}{$C, (i, m_i)_{i \in [N]}, (i, m'_i)_{i \in [N]}, \T$}
        \Let{b_0}{\bot}
        \Let{U}{\textsc{Empty}()}
        \For{$j \in [k]$}
            \Let{(b_1,\ldots, b_h)}{\bin(i_j)}
            \Let{\T[b_0,b_1,\ldots,b_h]}{H(m'_{i_j})}
            \Let{U[b_0,b_1,\ldots,b_h]}{\T[b_0,b_1,\ldots,b_h]}
            \For{$d = h-1, \ldots, 0$}
                \Let{\T[b_0,\ldots,b_d]}{H(\T[b_0,\ldots,b_d,b_{d+1}],\T[b_0,\ldots,b_d,b_{d+1}])}
                \Let{U[b_0,\ldots,b_d]}{\T[b_0,\ldots,b_d]}
            \EndFor
        \EndFor
        \Let{C'}{\T[b_0]}
        \State\Return $(C', U, \T)$
    \EndAlg
    \Alg{\sc ProofUpdate}{$C, \pi_x , m'_x, x, U$}
        \Let{\pi'_x}{\{\}}
        \Let{(b_1,\ldots,b_h)}{\bin(x)}
        \For{$d = h, \ldots, 1$}
            \If{$(b_0,b_1 \ldots\overline{b}_d) \in U$}
                \Let{\pi'_x[b_0,b_1 \ldots\overline{b}_d]}{U[b_0,b_1 \ldots\overline{b}_d]}
            \Else
                \Let{\pi'_x[b_0,b_1 \ldots\overline{b}_d]}{\pi_x[b_0,b_1 \ldots\overline{b}_d]}
            \EndIf
        \EndFor
    \State\Return $\pi'_x$
    \EndAlg
    \end{algorithmic}
\end{algorithm}

\end{document}